\pdfoutput=1

\documentclass{CSML}

\def\dOi{12(1:4)2016}
\lmcsheading%
{\dOi}
{1--34}
{}
{}
{Jan.~13, 2014}
{Apr.~27, 2016}
{}

\ACMCCS{[{\bf Theory of computation}]: Models of
  computation---Computability---Lambda calculus; Semantics and
  reasoning---Program semantics---Denotational semantics} 

\usepackage{cmll}
\usepackage{xspace}
\usepackage{multicol}
\usepackage{amsfonts}
\usepackage{amsmath}
\usepackage{amsthm}
\usepackage[disable]{todonotes}
\usepackage{hyperref}
\usepackage{etoolbox}
\usepackage{float}
\floatstyle{boxed} 
\restylefloat{figure}
\usepackage{framed}
\usepackage{txfonts}
\usepackage{graphicx}
\usepackage{caption}
\usepackage{subcaption}
\usepackage{amssymb}
\usepackage{bbold}
\usepackage{bussproofs}
\usepackage{tabularx}

\usepackage{pgf,tikz}
\usetikzlibrary{matrix,arrows}%,positioning,backgrounds,decorations,shadows,petri,cd,decorations.pathmorphing}
\makeatletter
\newcommand\mmodif[1]{\ifinalign@ #1 \else \ifmmode #1 \else $#1$\xspace \fi \fi} %macro to define macro that works in math and text mode alike
\makeatother

\newcommand\at{@} 
\newcommand\+{\mathrm{+}}                % sum relation
\newcommand\dash{\mathrm{\text{-}}}
\newcommand\id{\mmodif{\mathop{id}}}
\newcommand\FV{\mmodif{\mathrm{FV}}}     % Free variables
\newcommand\rta{\mmodif{\rightarrow}}
\newcommand\Rta{\mmodif{\Rightarrow}}

\newcommand\da{\mathrm{\downarrow}}
\newcommand\Da{\mathrm{\Downarrow^h}}

\newcommand\Ua{\mathrm{\Uparrow^h}}
\newcommand\cons{\mathrm{\rta}}

\newcommand\ruledarrow[3]{\stackrel{#1}{\rta}\!\!{}_{#2}^{#3}}

\newcommand\Lra{\Leftrightarrow}

\newcommand\llb{\llbracket}
\newcommand\rrb{\rrbracket}

\newcommand\llc{(\! |}
\newcommand\rrc{|\! )}

\newcommand\I{\mmodif{\boldsymbol{I}}}
\newcommand\J{\mmodif{\boldsymbol{J}}}
\newcommand\Id{\mmodif{\boldsymbol{I}}}

\newcommand\G{\mmodif{\boldsymbol{G}}}

\newcommand\Succ{\mmodif{\boldsymbol{S}}}
\newcommand\FixP{\mmodif{\boldsymbol{\Theta}}}
\newcommand\Om{\mmodif{\boldsymbol{\Omega}}}

\newcommand\Dinf{\mmodif{D_\infty}}
\newcommand\Pinf{\mmodif{P_{\!\infty}}}
\newcommand\Hst{\mmodif{\mathcal{H}^*}}

\newcommand\BT{\mmodif{\boldsymbol{BT}}}
\newcommand\BTf{\mmodif{\boldsymbol{BT_{\!\!f}}}}
\newcommand\BTqf{\mmodif{\boldsymbol{BT_{\!\!qf}}}}

\newcommand\BTomf{\mmodif{\boldsymbol{BT_{\!\!\Omega f}}}}
\newcommand\Nat{\mmodif{\mathbb{N}}}
\newcommand\Lamb{\mmodif{\Lambda}}

\newcommand\Achf[1]{\mmodif{\mathcal{A}_f(#1)}}     %Finite antichains
\newcommand\Comp[1]{\mmodif{\bar #1}}               %Extentional completion
\newcommand\Compl[1]{\mmodif{\overline{#1}}}        %idem
\newcommand\Scottl{\textsc{ScottL}\xspace}  %category ScottL
\newcommand\Scottlb{{\textsc{ScottL}$_!$}\xspace} %klesli category for ScottL

\newcommand\geetinf{\mmodif{\succeq_{\eta\infty}}}
\newcommand\leetinf{\mmodif{\preceq_{\eta\infty}}}
\newcommand\leet{\mmodif{\preceq_{\eta}}}

\newcommand\leob{\mmodif{\sqsubseteq_{\mathcal{H}^*}}}

\newcommand\equivob{\mmodif{\equiv_{\mathcal{H}^*}}}

\newcommand\Lcalcul{$\lambda$-calculus\xspace}
\newcommand\Lcalculs{$\lambda$-calculi\xspace}

\newcommand\Kweb{K-model\xspace}

\newcommand\Kwebs{K-models\xspace}

\newcommand\newsym[1]{\mmodif{#1}}%\index{sym}{$#1$|hyperpage}} %introduction of new symbols
\newcommand\newsyminv[1]{}%\index{sym}{$#1$|hyperpage}} %-
\newcommand\newsyminvinv[2]{}%\index{sym}{$#1$!#2|hyperpage}} %-
\newcommand\newsyminvsec[2]{}%\index{sym}{#1!$#2$|hyperpage}} %-
\newcommand\newsymprem[2]{\mmodif{#1}}%\index{sym}{$#1$!#2|hyperpage}} %-
\newcommand\newsymsec[2]{\mmodif{#2}}%\index{sym}{#1!$#2$|hyperpage}} %-
\newcommand\newsymsecinv[3]{\mmodif{#2}}%\index{sym}{#1!$#2$!#3|hyperpage}} %-
\newcommand\newsyminvinvinv[3]{}%\index{sym}{#1!$#2$!#3|hyperpage}} %-
\newcommand\newdef[1]{{\em #1}}%\index{def}{#1|hyperpage}} %introduction of new definitions
\newcommand\newdefinv[1]{}%\index{def}{#1|hyperpage}} %-
\newcommand\newdefinvinv[2]{}%\index{def}{#1!#2|hyperpage}} %-
\newcommand\newdefprem[2]{{\em #1}}%\index{def}{#1!#2|hyperpage}} %-
\newcommand\newdefpremsec[2]{{\em #1 #2}}%\index{def}{#1!#2|hyperpage}} %-
\newcommand\newdefsecprem[2]{{\em #2 #1}}%\index{def}{#1!#2|hyperpage}}  %-
\newcommand\newdeftheprem[3]{{\em #3 #1}}%\index{def}{#1!#2!#3|hyperpage}}  %
\newcommand\newdefthepremsec[3]{{\em #3 #1 #2}}%\index{def}{#1!#2!#3|hyperpage}}  %
\newcommand\newdefinvinvinv[3]{}%\index{def}{#1!#2!#3|hyperpage}}  %

\newenvironment{definition}{\begin{defi}}{\end{defi}}
\newenvironment{theorem}{\begin{thm}}{\end{thm}}
\newtheorem{lemma}[thm]{Lemma}
\newtheorem{example}[thm]{Example}
\newtheorem{remark}[thm]{Remark}

\newenvironment{corollary}{\begin{cor}}{\end{cor}}

\newenvironment{boite}{\vspace{-5pt}\begin{framed}\begin{center}\vspace{-10pt}}
                       {\end{center}\vspace{-10pt}\end{framed}\vspace{-5pt}}

\newenvironment{grammarLSB}{\begin{tabular}{l l r c l  r}}
                       {\end{tabular}}

\begin{document}

\title[On the characterization of models of \Hst]
      {On the characterization of models of \Hst: The semantical aspect}

\author[F.~Breuvart]{Flavien Breuvart}
 \address{PPS, UMR 7126, Univ Paris Diderot, LIPN, UMR 7030, Univ Paris Nord, Sorbonne Paris Cit\'e}
 \email{breuvart@pps.univ-paris-diderot.fr}
 \thanks{Partially founded by French ANR project COQUAS (number 12 JS02 006 01)}

\maketitle    

\begin{abstract}
We give a characterization, with respect to a large class of models of untyped $\lambda$-calculus, of those models that are fully abstract for head-normalization, {\em i.e.}, whose equational theory is $\mathcal{H}^*$ (observations for head normalization). An extensional K-model $D$ is fully abstract if and only if it is hyperimmune, {\em i.e.}, not well founded chains of elements of $D$ cannot be captured by any recursive function.

This article, together with its companion paper~\cite{Bre16SemPart} form the long version of \cite{Bre14}. It is a standalone paper that presents a purely semantical proof of the result as opposed to its companion paper that presents an independent and purely syntactical proof of the same result.
\end{abstract}

\section*{Introduction}
% !TEX root = main.tex
% 

The histories of full abstraction and denotational semantics
of \Lcalculs are both rooted in four fundamental articles published in
the course of one year. 

In 1976, Hyland \cite{Hyl75} and Wadsworth \cite{Wad76}
independently\footnote{The idea already appears
in Wadsworth thesis 3 years earlier.} proved the first full abstraction
result of Scott's $\Dinf$ (reflexive Scott's domain) for \Hst (observations for head normalization). The following year,
Milner \cite{Mil77} and Plotkin \cite{Plo77} showed respectively 
that PCF (a Turing-complete extension of the simply typed \Lcalcul) has a
unique fully abstract model up to isomorphism and that this model is
not in the category of Scott domains and continuous functions.

% (up to isomorphism and among a restricted class of domains)
% (the only available class at the time)

Later, various articles focused on circumventing Plotkin
counter-example \cite{AJM94,HyOn00} or investigating full abstraction
results for other calculi \cite{AbMc97,lai97,Pao06}. However, hardly anyone
pointed out the fact that Milner's uniqueness theorem is specific to
PCF, while \Hst has various models that are fully abstract but not
isomorphic.

%*** This was indeed a small issue from an applicative point of view,
%Milner characterization was a degenerated case and the general case
%did not seem that promising. However from a theoretical point of view
%the characterization of the full abstraction property for a non
%degenerated case still got some interest. *** 
%That is how we will studies the characterization of the full
%abstraction property for the $\lambda$-calculus, which is the
%non-degenerated case with the most reference material.
 
%The absence of the separation theorem on the $\lambda$-calculus, by
%allowing numerous fully abstract models, makes easier the research of
%such, but more difficult their characterization.

The quest for a general characterization of the fully abstract models
of head normalization started by successive refinements of a
sufficient, but not necessary condition \cite{DFH99,Gou95,Man09},
improving the proof techniques from
1976 \cite{Hyl75,Wad76}. %Nevertheless those results are quite
                         %interesting, especially since the last
                         %result \cite{Man09} applies for any class of
                         %models (contrary to Milner's theorem that is
                         %restricted to a specific class of domains),
                         %even including models with no enough points.
While these results shed some light on various fully abstract
semantics for \Hst, none of them could reach a full characterization.
% for good reason. Indeed these proufs are based on characterization
% of observational equivalence by Nakajima trees \cite{Nak75} that is
% coarser than ours. Coarser in the sens that Nakajima's trees
% introduce some infinite completions of terms that are not part of
% the lambda-calculus and that our models can distinguish from real
% terms. 

%Only Manzonetto's refinement \cite{Man09}, that also include models
%with no enough points, can attempt it: it seems that those models
%extensionally collapse into more usual ones \cite{Ehr09}. But the
%computational content of extensional collapse remaining relatively
%unknown it is difficult to conclude.

In this article, we give the first full characterization of the full
abstraction of an observational semantics for a specific (but large)
class of models. 
%, in the same manner as Milner restricted himself. 
%Such restriction offer finer understanding of behavior of our models,
%which is especially useful to reach the necessary condition. However
%doing so is dangerous since the class has to be sufficiently
%expressive to have precise non degenerated results. Moreover the
%chosen class of model has better to simplify the proof as much as
%possible (** since our goal is to understand behavior of our property
%**).
The class we choose is that of Krivine-models, \linebreak
or K-models \cite{Krivine,Ber00}. 
This class, described in Section~\ref{ssec:K-models}, is essentially
the subclass of Scott complete lattices (or filter
models \cite{CDHL84}) which are prime algebraic. We add two further
conditions: extensionality and approximability of
Definition~\ref{def:approximationProp}. 
Extensionality is a standard and perfectly understood notion that requires the model to respect the $\eta$-equivalence. Notice that it is a necessary condition for the full abstraction if \Hst. Approximability is another standard notion saying that the model reflects the fact that a term is approximated by its finite B\"ohm trees. This notion has been extensively studied \cite[Section~III.17.3]{BarDekSta}.

The extensional and
approximable \Kwebs are the objects of our characterization and can be
seen as a natural class of models obtained from models of linear
logic \cite{Gir87}. Indeed, the extensional \Kwebs correspond to the
extensional reflexive objects of the co-Kleisli
category associated with the exponential comonad of Ehrhard's \Scottl
category \cite{Ehr09} (Prop.~\ref{prop:reflOb}).
 %But we are particuliary interested in the extensional subclass since
 %extensionality is a necessary condition for full abstraction. This
 %subclass can then be seen as a natural class of models obtained from
 %models of linear logic. Indeed it corresponds to extensional
 %reflexive objects of the category \Scottlb, {\em i.e.}, the
 %co-Kleisli construction obtained from
 %Ehrhard's \Scottl \cite{Ehr09}.

%--- One can remark that all relational models of the
%$\lambda$-calculus -REF!!- commute with B\"ohm Trees (since their
%interpretations commute with Taylor expansion -REF!!- itself
%commuting with B\"ohm Trees -REF-). However relational models would
%also add undesired other technical subtilities in the proofs.

We achieve the characterization of full abstraction for \Hst in
Theorem~\ref{th:final}: a model $D$ is fully abstract for \Hst iff $D$
is {\em hyperimmune} (Def.~\ref{def:hyperim}). Hyperimmunity is the
key property our study introduces in denotational semantics. This
property is reminiscent of the Post's notion of hyperimmune sets in
recursion theory. Hyperimmunity in recursion theory is not only undecidable, but also
surprisingly high in the hierarchy of undecidable properties (it
cannot be decided by a machine with an oracle deciding the halting
problem) \cite{Nies}.

Roughly speaking, a model $D$ is hyperimmune whenever the
$\lambda$-terms can have access to only well-founded\footnote{well-foundedness
is considered with regard to a new order independent from the poset order of $D$.}
chains of elements of $D$. In other words, $D$ might have non-well-founded
\linebreak chains $d_0\ge d_1\ge\cdots$, but these chains ``grow'' so fast (for a
suitable notion of growth), that they cannot be contained in the
interpretation of any $\lambda$-term.

The intuition that full abstraction of \Hst is related to a kind of
well-foundedness can be found in the literature ({\em e.g.},
Hyland's \cite{Hyl75}, Gouy's \cite{Gou95} or
Manzonetto's \cite{Man09}). Our contribution is to give, with
hyperimmunity, a precise definition of this intuition, at least in the
setting of \Kwebs.

A finer intuition can be described in terms of game
semantics. Informally, a game semantics for the untyped \Lcalcul takes
place in the arena interpreting the recursive type $o = o\rta o$. This
arena is infinitely wide (by developing the left $o$) and infinitely
deep (by developing the right~$o$). Moves therein can thus be
characterized by their nature (question or answer) and by a word over
natural numbers. For example,~$q(2.3.1)$ represents a question in the
underlined~``$o$'' in~$o=o\cons (o\cons o\cons (\underline o \cons
o)\cons o)\cons o$. Plays in this game are potentially infinite
sequences of moves, where a question of the form $q(w)$ is followed by
any number of deeper questions/answers, before an answer $a(w)$ is
eventually provided, if any.
%A play in such a game is an possibly infinite stream (representing a
%function with non bounded behavior) of moves where questions $q(w)$
%are followed by any number (even infinite) of deeper questions
%$q(w.n_1)...q(w.n_k)$ before their answer $a(w)$ (and then the game
%can continue by new questions or answers).

A play like $q(\epsilon),q(1)...a(1),q(2)...a(2),q(3)...$ is
admissible: one player keeps asking questions and is infinitely
delaying the answer to the initial question, but some answers are
given so that the stream is productive. However, the full abstraction
for \Hst forbids non-productive infinite questioning 
%without any answers by anyone, {\em i.e.}, the situation where
%opponents and player continuously ask a question deeper in the type
%rather than answering like the in play 
like in $q(\epsilon),q(1),q(1.1),q(1.1.1)...$, in general. 
%The concept of hyperimmunity is going further than that by allowing
%such plays given that the questioning is {\em hyperimmune}, {\em
%i.e.}, that the function that choose the range of questioned
%arguments at each step is not majored by any recursive function. 
Nevertheless, disallowing {\em all} such strategies is sufficient, but
not necessary to get full abstraction. The hyperimmunity condition is
finer: non productive infinite questioning is allowed {\em as long as}
the function that chooses the next question grows faster than any
recursive function (notice that in the example above that choice is
performed by the constant $(n\mapsto 1)$ function). For example, if
$(u_i)_{i\ge 0}$ grows faster than any recursive function, the play
$q(\epsilon),q(u_1),q(u_1.u_2),q(u_1.u_2.u_3)...$ is perfectly
allowed. % In fact, such plays are allowed because they are not
         % accessible if the strategies come from recursive terms.

%(Definition~\ref{def:hyperim}). Indeed, hyperimmunity is a
%simple-to-express property that is not only undecidable, but also
%surprisingly high in the hierarchy of undecidable properties (it
%cannot be decided by a machine with an oracle deciding the halting
%problem) \cite{Nies}. It roughly express that computable terms have
%only access to a well founded fragment of the model.

%Intuitively, if we see elements of the models $D$ as infinite stacks
%of sets over $D$, hyperimmunity asks that any infinite chain
%recursively picking any element appearing in the description of the
%previous, has to explore the successive stacks deeper and deeper by
%following a function that grows faster than any computable one. This
%result suggests that a fine analysis of the relations between models
%and syntax cannot avoid technical computability issues.

Incidentally, we obtain a significant corollary (also expressed in
Theorem~\ref{th:final}) stating that full abstraction coincides with
inequational full abstraction for \Hst (equivalence between
observational and denotational orders). This is in contrast to what
happens to other calculi \cite{Sto90,EPT14}. % \cite{EPT14}. We
                                % cannot state that it is a general
                                % fact rather that an artefact coming
                                % from the choice of class of
                                % models. However we are pretty
                                % confident that this result is
                                % unexpected and covers at least a
                                % large class of models.

In the literature, most of the proofs of full abstraction for \Hst are
based on Nakajima trees \cite{Nak75} or some other notion of quotient
of the space of B\"ohm trees, using the characterization of the
observational equivalence (see Proposition~\ref{prop:H*le}). The usual
approach is too coarse because it considers arbitrary B\"ohm trees
which are not necessarily images of actual $\lambda$-terms. To
overcome this we propose two different techniques leading to two
different proofs of the main result: one purely semantical and the
other purely syntactical. In this article we only present the former, 
the latter being the object of a companion paper \cite{Bre16SemPart}.

This proof
follows the line of historical ones while overcoming weaknesses of
Nakajima trees with a notion of quasi-approximation property
(Def.~\ref{def:quasi-approxTh}), that involves recursivity in a
refined way. Quasi-approximability is a key tool in the proof, which
is otherwise quite standard.  However, since B\"ohm trees are specific
to the \Lcalcul and head reduction, there is not much hope to extend
the proof to many other calculi/strategies (such as
differential \Lcalcul \cite{EhRe04}, or call-by-value strategies).

\section{Preliminaries and result}
    \label{sec:Theorem}

  \subsection{Preliminaries}
     \label{ssec:preliminaries}
\subsubsection{Preorders} \label{PrelA}\label{sec:posets} \quad \newline%
Given two partially ordered sets $D=(|D|,\le_D)$ and $E=(|E|,\le_E)$, we denote:
\begin{itemize}
  \item $D^{op}=(|D|,\ge_D)$ the reverse-ordered set.
  \item $D\times E=(|D|\times |E|,\le_{D\times E})$ the Cartesian product endowed with the pointwise order: 
    $$(\delta,\epsilon)\le_{D\times E}(\delta',\epsilon')\quad \text{if}\quad \delta\le_D\delta'\quad \text{and}\quad \epsilon\le_E\epsilon'.$$
  \item $\mathcal{A}_f(D)=(|\mathcal{A}_f(D)|,\le_{\mathcal{A}_f(D)})$ the set of finite antichains of $D$ ({\em i.e.}, finite subsets whose elements are pairwise incomparable) endowed with the order :
 $$a\le_{\mathcal{A}_f(D)} b\ \Lra\ \forall \alpha\in a,\exists\beta\in b, \alpha\le_D\beta$$
\end{itemize}\medskip

\noindent In the following will we use $D$ for $|D|$ when there is no ambiguity. Initial Greek letters $\alpha,\beta,\gamma...$ will vary on elements of ordered sets. Capital initial Latin letters $A,B,C...$ will vary over subsets of ordered sets. And finally, initial Latin letters $a,b,c...$ will denote finite antichains. 

An {\em order isomorphism} between $D$ and $E$ is a bijection $\phi:|D|\rta|E|$ such that $\phi$ and $\phi^{-1}$ are monotone. 

Given a subset $A\subseteq |D|$, we denote $\mathrm{\downarrow} A=\{\alpha\mid\exists\beta\in A,\alpha\mathrm{\le}\beta\}$. We denote by $\mathcal{I}(D)$ the set of {\em initial segments of $D$}, that is  $\mathcal{I}(D)=\{\mathrm{\downarrow}A\mid A\subseteq |D|\}$.
The set $\mathcal{I}(D)$ is a prime algebraic complete lattice with respect to the set-theoretical inclusion. The {\em sups} are given by the unions and the {\em prime elements} are the downward closure of the singletons. The {\em compact elements} are the downward closure of finite antichains. 

The domain of a partial function $f$ is denoted by $Dom(f)$. The {\em graph} of a Scott-continuous function $f:\mathcal{I}(D)\rta \mathcal{I}(E)$ is
\begin{equation} 
  \mathrm{graph}(f)=\{(a,\alpha)\in \mathcal{A}_f(D)^{op}\mathrm{\times}E\mid \alpha\in f(\mathrm{\downarrow} a)\} \label{eq:graphf}
\end{equation}
Notice that elements of $\mathcal{I}(\mathcal{A}_f(D)^{op}\mathrm{\times}E)$ are in one-to-one correspondence with the graphs of Scott-continuous functions from $\mathcal{I}(D)$ to $\mathcal{I}(E)$.

\subsubsection{\Lcalcul} \quad \newline%
The $\lambda$-terms are defined up to $\alpha$-equivalence by the following grammar using notation ``{\em \`a la} Barendregt'' \cite{Barendregt} (where variables are denoted by final Latin letters $x,y,z...$): 
\begin{center}
\begin{tabular}{c l c c l}
$\Lamb$ & ($\lambda$-terms) & $M,N$ & $::=$ & $x\quad |\quad \lambda x.M\quad |\quad M\ N$
\end{tabular}
\end{center}
We denote by $\FV(M)$ the set of free variables of a $\lambda$-term $M$. Moreover, we abbreviate a nested abstraction $\lambda x_1...x_k.M$ to $\lambda\vec x^{\, k\!} M$, or, when $k$ is irrelevant, to $\lambda\vec x M$. We denote by $M[N/x]$ the capture-free substitution of $x$ by $N$.

\noindent The $\lambda$-terms are subject to the $\beta$-reduction:
\[
 (\beta)\quad\quad (\lambda x.M)\ N\  \ruledarrow{\beta}{}{}\ M[N/x]
\]
A context $C$ is a $\lambda$-term with possibly some occurrences of a hole, {\em i.e.}:
\begin{center}
\begin{tabular}{c l c c l}
  $\Lamb^{\llc.\rrc}$ & (contexts) &  $C$ & $::=$ & $\llc.\rrc\quad |\quad x\quad |\quad \lambda x.C\quad |\quad C_1\ C_2$
\end{tabular}
\end{center}
The writing $C\llc M\rrc$ denotes the term obtained by filling the holes of $C$ by $M$. 
The small step reduction~$\rta$ is the closure of $(\beta)$ by any context, and $\rta_h$ is the closure of $(\beta)$ by the rules:
\begin{center}
  \AxiomC{$M\rta_h M'$}
  \UnaryInfC{$\lambda x.M\rta_h \lambda x.M'$}
  \DisplayProof\hskip 10pt
  \AxiomC{$M\rta_h M'$}
  \AxiomC{\hspace{-3pt}$M$ is an application}
  \BinaryInfC{$M\ N\rta_h M'\ N$}
  \DisplayProof
\end{center}
The transitive reduction $\rta^*$ (resp $\rta_h^*$) is the reflexive transitive closure of $\rta$ (resp $\rta_h$).\\ The big step head reduction, denoted $M\Da N$, is $M\rta_h^*N$ for $N$ in a {\em head-normal form},\linebreak {\em i.e.}, $N=\lambda x_1...x_m.y\ M_1\cdots M_n$, for $M_1,...,M_m$ any terms. We write $M\Da$ for the ({\em head}) {\em convergence}, {\em i.e.}, whenever there is $N$ such that $M\Da N$. 

\begin{example}
  \begin{itemize}
  \item The \newdef{identity term} $\newsym{\protect\I}:=\lambda x.x$ takes a term and returns it as it is:
    $$ \I\ M\quad \rta\quad M.$$
  \item The $n^{th}$ Church numeral, denoted by \newsym{\protect\underline{n}}, and the successor function, denoted by \newsym{\protect\boldsymbol{S}}, are defined by 
    \begin{align*}
      \underline{n} &:= \lambda fx.\underbrace{f\ (f \cdots\ f\ (f}_{n\ \text{times}}\ x)\cdots ), & \boldsymbol{S}&:= \lambda ufx. u\ f\ (f\ x).
    \end{align*}
    Together they provide a suitable encoding for natural numbers, with $\underline n$ representing the $n^{th}$ iteration.
  \item The \newdef{looping term} $\newsym{\protect\Om}:= (\lambda x.xx)\ (\lambda x.xx)$ infinitely reduces into itself, notice that $\Om$ is an example of a diverging term:
    $$ \Om \quad \rta \quad (x\ x)[\lambda y.y\ y/x]\quad = \quad \Om \quad \rta \quad \Om \quad \rta \quad \cdots.$$
  \item The \newdefsecprem{combinator}{Turing fixpoint} $\newsym{\protect\FixP} := (\lambda uv.v\ (u\ u\ v))\ (\lambda uv.v\ (u\ u\ v))$ is a term that computes the least fixpoint of its argument (if it exists):
    \begin{align*}
      \FixP\ M &\rta (\lambda v.v\ ((\lambda uv.v\ (u\ u\ v))\ (\lambda uv.v\ (u\ u\ v)) v))\ M \\
      &= (\lambda v.v\ (\FixP\ v))\ M \\
      &\rta M\ (\FixP\ M).
    \end{align*}
%  \item The \newdefthepremsec{$\eta\infty$-expansion}{of the identity}{main} $\newsym{\protect\J} := \FixP\ (\lambda uxy. x\ (u\ y)))$ is a term that somehow ``converges to the identity'':
%    $$ \J \quad \rta^*\quad \lambda z x_1. z\ (\J\ x_1) \quad \rta^*\quad \lambda z x_1. z\ (\lambda x_2. x_1\ (J\ x_2)) \quad \rta^*\quad \lambda z x_1. z\ (\lambda x_2. x_1\ (\lambda x_3.x_2\ (J\ x_3))). $$
  \end{itemize}
\end{example}

\noindent Other notions of convergence exist (strong, lazy, call by value...), but our study focuses on head convergence, inducing the equational theory denoted by \Hst. 

\begin{definition}\label{observationalPreorder}  
  The {\em observational preorder} and {\em equivalence} denoted $\leob$ and $\equivob$ are given by:
  \begin{align*}
    M &\leob N  & \text{if}\quad & & \forall C ,\ C\llc M\rrc\Da\ \Rta\ C\llc N\rrc\Da,\\
    M&\equivob N & \text{if}\quad & & M\leob N \text{ and } N\leob M.
  \end{align*}
  The resulting (in)equational theory is called \Hst.
\end{definition}

\begin{definition}
  A model of the untyped \Lcalcul with an interpretation $\llb-\rrb$ is:
  \begin{itemize}
  \item fully abstract (for \Hst) if for all $M, N\in\Lamb$:
    $$M \equivob N \quad \text{if} \quad \llb M\rrb = \llb N\rrb,$$
  \item inequationally fully abstract (for \Hst) if for all $M, N\in\Lamb$:\footnote{This can be generalised by replacing $\subseteq$ by any order on the model.}
    $$M \leob N \quad \text{if} \quad \llb M\rrb \subseteq \llb N\rrb.$$
  \end{itemize}
\end{definition}

\noindent Henceforth, convergence of a $\lambda$-term means head convergence, and full abstraction for \Lcalcul means full abstraction for \Hst.

Concerning recursive properties of \Lcalcul, we will use the following one:
  \begin{prop}[{\cite[Proposition~8.2.2]{Barendregt} \footnotemark}]\label{prop:autorecursivity}
    \quad\\
    Let $(M_n)_{n\in \Nat}$ be a sequence of terms such that:
    \begin{itemize}
    \item $\forall n\in \Nat, M_n\in \Lamb^0$,
    \item the encoding of $(n\mapsto M_n)$ is recursive,
    \end{itemize}
    then there exists $F$ such that:
    $$\forall n, F\ \underline{n}\ \rta^* M_n.$$
  \end{prop}
  \footnotetext{This is not the original statement. We remove the dependence on $\vec x$ that is empty in our case and we replace the $\beta$-equivalence by a reduction since the proof of Barendregt \cite{Barendregt} works as well with this refinement.}

\subsection{K-models} \quad \newline\label{ssec:K-models}%
% !TEX root = main.tex
% 
%
We introduce here the main semantical object of this article: extensional \Kwebs \cite{Krivine}\cite{Ber00}. This class of models of the untyped \Lcalcul is a subclass of filter models \todo{define filter models?}\cite{CDHL84} containing many extensional models from the continuous semantics, like Scott's $\Dinf$ \cite{Scott}.

\subsubsection{The category \protect\Scottlb} \quad \newline%
Extensional \Kwebs %, through stone duality, 
correspond to the extensional reflexive Scott domains that are prime algebraic complete lattices and whose application embeds prime elements into prime elements \cite{Hut93,Win98}. However we prefer to exhibit \Kwebs as the extensional reflexive objects of the category \Scottlb which is itself the Kleisli category over the linear category \Scottl \cite{Ehr09}.

\begin{definition}
We define the Cartesian closed category \newdef{\protect\Scottlb} \cite{Hut93,Win98,Ehr09}:
\begin{itemize}
\item {\em objects} are partially ordered sets.
\item {\em morphisms} from $D$ to $E$ are Scott-continuous functions between the complete lattices $\mathcal{I}(D)$ and $\mathcal{I}(E)$.
\end{itemize}
The {\em Cartesian product} is the disjoint sum of posets. The {\em terminal object} $\top$ is the empty poset. The {\em exponential object} $D\mathrm{\Rta} E$ is $\Achf{D}^{op}\mathrm{\times} E$. Notice that an element of $\mathcal{I}(D\mathrm{\Rta} E)$ is the graph of a morphism from $D$ to $E$ (see Equation \eqref{eq:graphf}). This construction provides a natural isomorphism between $\mathcal{I}(D\mathrm{\Rta} E)$ and the corresponding homset. Notice that if $\simeq$ denotes isomorphisms in \Scottlb, then:
\begin{equation}
  D\Rta D\Rta \cdots \Rta D\simeq (\Achf{D}^{op})^n\times D. \label{eq:defRta}
\end{equation}
For example $D\Rta(D\Rta D)\simeq \Achf{D}^{op}\times(\Achf{D}^{op}\times D)= (\Achf{D}^{op})^2\times D$.
\end{definition}

%Our goal is give an abstract characterization of \fuAb. 
%We will restrict our study to % a class of model that is as simple as possible while preserving the essence of the problem. We do not presume that our choice is the best one, but it was sufficient to perform a proof for a non-trivial characterization.
%extentional \Kwebs \cite{Krivine}\cite{Ber00}, which constitute a subclass of filter models \cite{CDHL84} containing many extensional models from the continuous semantics, like Scott's $\Dinf$ \cite{Scott}. K-models correspond to the reflexive Scott domains that are prime algebraic complete lattices (and whose application embeds prime elements into prime elements). %We will restrict ourselves to extensional \Kwebs, which are the K-models enjoying $\eta$-equivalence.% this is easier and sufficient for our goal of characterizing full abstraction. 
%Finally, later on we will add an other restriction, the commutation with Böhm trees, that is a natural restriction that preserves most of the extensional sensible \Kwebs but forbid some gruesome ones (if any).

\begin{remark}\label{rk:equivAchSubset}
  In the literature ({\em e.g.} \cite{Hut93,Win98,Ehr09}), objects are preodered sets and the exponential object $D\Rta D$ is defined by using finite subsets (or multisets) instead of the finite antichains. Our presentation is the quotient of the usual one by the equivalence relation induced by the preorder. The two presentations are equivalent (in terms of equivalence of category) but our choice simplifies the definition of hyperimmunity (Definition~\ref{def:hyperim}).
\end{remark}

\begin{prop} 
  The category \Scottlb is isomorphic to the category of prime algebraic complete lattices and Scott-continuous maps.
\end{prop}
\begin{proof}
  Given a poset $D$, the initial segments $\mathcal{I}(D)$ form a prime algebraic complete lattice which prime elements are the downward closures ${\downarrow}\alpha$ of any $\alpha\in D$ since $I=\bigcup_{\alpha\in I}{\downarrow}\alpha$. Conversely, the prime elements of a prime algebraic complete lattice form a poset. The two operations are inverse one to the other modulo \Scottlb-isomorphisms and Scott-continuous isomorphisms.
\end{proof}

\subsubsection{An algebraic presentation of K-models}

\begin{definition}[\cite{Krivine}]\label{def:K-model}
  An \newdef{extensional K-model} is a pair $(D,i_D)$ where:
  \begin{itemize}
  \item $D$ is a poset.
  \item \newsym{\protect i_D} is an order isomorphism between $D\mathrm{\Rta} D$ and $D$.
  \end{itemize}
\end{definition}

By abuse of notation we may denote the pair $(D,i_D)$ simply by $D$ when it is clear from the context we are referring to an extensional \Kweb. 

%\begin{remark}
%  As a reminiscence of Remark~\ref{rk:equivAchSubset}, the literature ({\em e.g.} \cite{Ber00}), generally consider preordered sets and subsets (or multisets) are used in place of antichains. This time too, both presentations are equivalent.
%\end{remark}

%\smallskip

\begin{definition}
  Given a Cartesian closed category $\mathcal C$, an extensional reflexive objects of $\mathcal C$ is an objects $D$ endowed with an isomorphism $abs_D:(D\Rta D)\rta D$  (and $app_D:=abs_D^{-1}$). This corresponds to the categorical axiomatisation of extensional models of the untyped \Lcalcul.
\end{definition}

\begin{prop} \label{prop:reflOb}
  Extensional \Kwebs correspond exactly to extensional reflexive objects of \Scottlb.
\end{prop}
\begin{proof}
  Given a \Kweb $(D,i_D)$, the isomorphism between $D\mathrm{\Rta}D$ and $D$ is given by:
  \begin{align*}
    \forall A &\in \mathcal{I}(D\mathrm{\Rta}D),  &\mathrm{app}_D(A) &= \{i_D(a,\alpha)\mid (a,\alpha)\in A\},\\
    \forall B &\in \mathcal{I}(D), &\mathrm{abs}_D(B) &= \{(a,\alpha)\mid i_D(a,\alpha)\in B\}.
  \end{align*}
  Conversely, consider an extensional reflexive object $(D,app_D,abs_D)$ of \Scottlb. Since $abs_D$ is an isomorphism, it is linear (that is, it preserves all sups). For all $(a,\alpha)\in D\mathrm{\Rta}D$, we have 
  $$\da(a,\alpha)=abs(app(\da(a,\alpha)))=\bigcup_{\beta\in app(\da(a,\alpha))}abs(\da\beta).$$
  Thus there is $\beta\in app(\da(a,\alpha)$ such that $(a,\alpha)\in abs(\da \beta)$, and since $abs(\da \beta)\subseteq\da(a,\alpha)$, this is an equality. Thus there is a unique $\beta$ such that $app_D(a,\alpha)=\da\beta$, this is $i_D(a,\alpha)$.
\end{proof}
%\begin{proof}
%  The left to right side is obtained by setting $abs_D(I)=\{i_D(a,\alpha)|(a,\alpha)\in I\}$. For the other side we verify that $abs_D(\da(a,\alpha))=\da\beta$ for some $\beta$.
%\end{proof}

In the following we will not distinguish between a \Kweb and its associated reflexive object, this is a model of the pure $\lambda$-calculus.

%\subsubsection{Partial \Kweb}

\begin{definition}\label{def:partialKweb}
An \newdef{extensional partial K-model} is a pair $(E,j_E)$ where $E$ is an object of \Scottlb and $j_E$ is a partial function from $E\mathrm{\Rta}E$ to $E$ that is an order isomorphism between $\mathrm{Dom}(j_E)$ and $E$.
$$ E\quad \stackrel{j_E}{\longleftrightarrow}\quad \mathrm{Dom}(j_E)\quad \subseteq\quad (E\Rta E) $$
\end{definition}
%Categorically speaking, using the fact that \Scottl is DCPO-enriched with the inclusion order (up-to isomorphisms), an extensional partial \Kwebs corresponds to an object $E$ isomorph to some $E'\subseteq (E\Rta E)$. 

\begin{definition}\label{def:Compl}
  The \newdefpremsec{completion}{of a partial K-model} $(E,j_E)$ is the union 
  $$(\Comp{E},j_{\Comp{E}})=(\bigcup_{n\in\mathbb{N}}E_n,\bigcup_{n\in\mathbb{N}}j_{E_n})$$
  of partial completions $(E_n,j_{E_n})$ that are extensional partial K-models defined by induction on $n$.\linebreak
  We define $(E_0,j_{E_0}):=(E,j_E)$ and:
  \begin{itemize}
  \item $|E_{n+1}|\ :=\ |E_n|\cup (|E_n\Rta E_n|-Dom(j_{E_n}))$
  \item $j_{E_{n+1}}$ is defined only over $|E_n\Rta E_n|\subseteq |E_{n+1}\Rta E_{n+1}|$ by $j_{E_{n+1}} := j_{E_n}\cup id_{|E_n\Rta E_n|-Dom(j_{E_n})}$
  \item $\le_{E_{n+1}}$ is given by $j_{E_{n+1}}(a,\alpha)\le_{E_{n+1}}(b,\beta)$ if $a\ge_{\Achf{E_n}} b$ and $\alpha\le_{E_n}\beta$.
%    \begin{itemize}
%    \item $(a,\alpha)\le_{E_{n+1}}(b,\beta)$ if $\alpha\le_{E_n}\beta$, $a\sqsupseteq_{E_n}b$, $(a,\alpha),(b,\beta)\not\in Dom(j_{E_n})$;
%    \item $(a,\alpha)\le_{E_{n+1}}j_{E_n}(b,\beta)$ if $\alpha\le_{E_n}\beta$, $a\sqsupseteq_{E_n}b$, $(a,\alpha)\not\in Dom(j_{E_n})$ and $(b,\beta)\in Dom(j_{E_n})$;
%    \item $j_{E_n}(a,\alpha)\le_{E_{n+1}}(b,\beta)$ if $\alpha\le_{E_n}\beta$, $a\sqsupseteq_{E_n}b$, $(a,\alpha)\in Dom(j_{E_n})$ and $(b,\beta)\not\in Dom(j_{E_n})$;
%    \item $j_{E_n}(a,\alpha)\le_{E_{n+1}}j_{E_n}(b,\beta)$ if $j_{E_n}(a,\alpha)\le_{E_{n}}j_{E_n}(b,\beta)$ and $(a,\alpha),(b,\beta)\in Dom(j_{E_n})$.
%    \end{itemize}
  \end{itemize}
\end{definition}
Remark that $E_{n+1}$ corresponds to $E_n\Rta E_n$ up to isomorphism, what leads to the equivalent definition:

\begin{prop}
  The completion $(\Comp{E},j_{\Comp{E}})$ of an extensional partial \Kweb $(E,j_E)$ can be described as the categorical $\omega$-colimit (in \Scottl) of $(E'_n)_n$ along the injections $(j_n^{-1})_n$. The posets $(E'_n)_n$ and the partial functions $(j_n)_n$ are defined by induction by $(E'_0,j_0):=(E,j_E)$, and for $n\ge 0$, by $E_{n+1}':=E_n'\Rta E_n'$ and for all $a\subseteq dom(j_n)$ and $\alpha\in j_n$,  $j_{n+1}(a,\alpha):=(j_{n}(a),j_{n}(\alpha))$.

  \begin{tikzpicture}[description/.style={fill=white,inner sep=2pt},ampersand replacement=\&]
    \matrix (m) [matrix of math nodes, row sep=2em, column sep=2em, text height=1.5ex, text depth=0.25ex]
    { \& \& \overline{E} \\
      E \& E_1 \& E_2 \& \cdots \& E_n \& \cdots \\
    };
    \path[->] 
      (m-2-1) edge[] node[below] {$j_E^{-1}$} (m-2-2)
      (m-2-2) edge[] node[below] {$j_{1}^{-1}$} (m-2-3)
      (m-2-3) edge[] node[below] {$j_{2}^{-1}$} (m-2-4)
      (m-2-4) edge[] node[below] {$j_{n-1}^{-1}$} (m-2-5)
      (m-2-5) edge[] node[below] {$j_{n}^{-1}$} (m-2-6)
      (m-2-1) edge[] node[auto]  {} (m-1-3)
      (m-2-2) edge[] node[auto]  {} (m-1-3)
      (m-2-3) edge[] node[right] {} (m-1-3)
      (m-2-5) edge[] node[right] {} (m-1-3);
  \end{tikzpicture}
  
\end{prop}
\todo{removable}

\begin{remark}\label{rk:completion}
The completion of an extensional partial \Kweb $(E,j_E)$ is the smallest extensional~\Kweb $\Comp{E}$ containing $E$. In particular, any extensional \Kweb $D$ is the extensional completion of itself: $D=\Comp{D}$.
\end{remark}
%Categorically speaking this is the least fixpoint of the functor $F:E\rta(E\Rta E)$ ..
%\begin{proof} 
%  One can easily check that $(\Compl{E},j_{\bar E})$ is an extensional total \Kweb. Moreover, $(\Compl{E},j_{\bar E})$ is the smallest one containing  $(E,j_E)$ since we have added only necessary elements.
%%    In order to distinguish it to the actual couple, if $(a,\alpha)\in \Comp{E}$ was added during the completion, we will denote it $a\cons\alpha$.
%\end{proof}

\begin{example}\label{example:1}\quad
\begin{enumerate}
  \item \label{enum:Dinf} \newdef{Scott's $\protect\Dinf$} \newsyminvsec{K-model}{\protect\Dinf} \cite{Scott} is the extensional completion of 
    \begin{align*}
      |D| &:= \{*\},   &   \le_D &:= \id,   &   j_D &:= \{(\emptyset,*)\mapsto *\}.
    \end{align*}
    The completion the a triple $(|\Dinf|,\le_{\Dinf},j_{\Dinf})$ where $|\Dinf|$ is generated by:
    \begin{center}
    \begin{grammarLSB}
      & $|\Dinf|$ &   $\alpha,\beta$ & $::=$ & $*\quad |\quad a\cons\alpha$\\
      & $|!\Dinf|$ &  $a,b$ & $\in$ & $\quad \Achf{|\Dinf|}$
    \end{grammarLSB}
    \end{center}
    except that $\emptyset\cons*\not\in|\Dinf|$; $j_{\Dinf}$ is defined by $j_{\Dinf}(\emptyset,*)=*$ and $j_{\Dinf}(a,\alpha)=a\cons\alpha$ \linebreak for $(a,\alpha)\neq(\emptyset,*)$.%; and $\le_{\Dinf}$ is inductively defined.% using formulas of Section~\ref{sec:posets}.
%    Where $\le_{D_\infty}$ is defined inductively by $\alpha\le_{D_\infty} *$ for all $\alpha$ and $a\mathrm{::}\alpha\le_{D_\infty} b\mathrm{::}\beta$ when $\alpha\le_{D_\infty}\beta$ and there is, for all $\beta'\in b$, $\alpha'\in a$ such that $\beta'\le_{D_\infty}\alpha'$.\\
  \item \label{enum:Park} \newdef{Park's $\protect\Pinf$} \newsyminvsec{K-model}{\protect\Pinf} \cite{Par76} is the extensional completion of
    \begin{align*}
      |P| &:= \{*\},   &   \le_P &:= \id,   &   j_P &:= \{(\{*\},*)\mapsto *\};
    \end{align*}
    {\em i.e.},  $|P_\infty|$ is defined by the previous grammar except that $(\{*\}\cons *)\not\in |P_\infty|$ while $\emptyset\cons *\in|P_\infty|$. 
  \item \label{enum:D*} \newsymsec{K-model}{Norm} or \newsymsec{K-model}{\protect D^*_\infty} \cite{CDZ87} is the extensional completion of
    \begin{align*}
      |E| &:= \{p,q\},   &   \le_E &:= \id\cup\{p<q\},
    \end{align*}
    \vspace{-2em}
    \begin{align*} 
      j_E &:= \{(\{p\},q)\mathrm{\mapsto} q , (\{q\},p)\mathrm{\mapsto} p\}.
    \end{align*}
  \item \label{enum:WellStrat} \newdef{Well-stratified \Kwebs} \cite{Man09} are the extensional completions of some $E$ respecting 
    $$\forall (a,\alpha)\mathrm{\in} \mathrm{Dom}(j_E), a\mathrm{=}\emptyset.$$
    % \item $\Compl{\mathbb{Z}}$ is the extensional completion of $|E|=\mathbb{Z}$, $\le_E=\id$ and $j_E=\{(\emptyset,n)\mathrm{\mapsto} n\+1\ |\ n\in \mathbb{Z}\}$
  \item \label{enum:Ind} The \newdef{inductive \Compl{\omega}} is the extensional completion of 
    \begin{align*}
      |E| &:= \mathbb{N},   &   \le_E &:= \id,   &   j_E &:= \{(\{k\mid k< n\},n)\mathrm{\mapsto} n\mid n\in \mathbb{N}\}.
    \end{align*}
  \item \label{enum:CoInd} The \newdef{co-inductive \Compl{\mathbb{Z}}} is the extensional completion of 
    \begin{align*}
      |E| &:= \mathbb{Z},   &   \le_E &:= \id,   &   j_E &:= \{(\{n\},\:n\+1)\mathrm{\mapsto} n\+1\mid n\in \mathbb{Z}\}.
    \end{align*}
  \item \label{enum:Func} \newdef{Functionals $H^f$} \newsyminvsec{K-model}{\protect H^f} (given $f:\Nat\rta\Nat$) are the extensional completions of:
    \begin{align*}
      |E|  &:= \{*\}\cup\{\alpha_j^n \mid n\ge 0,\ 1\le j\le f(n)\},   &
      \le_E &:= \id,
    \end{align*}
    \vspace{-1em}
    \[
      \quad\quad j_E := \ \Big\{(\emptyset, *)\mapsto*\Big\}\ \cup\ \Big\{(\emptyset, \alpha_{j+1}^n)\mapsto\alpha_j^n\mid 1\le j< f(n)\Big\}\ \cup\ \Big\{(\{\alpha_{1}^{n+1}\},*)\mapsto\alpha_{f(n)}^n\mid n\in\Nat^*\Big\},
    \]
    where $(\alpha^n_j)_{n,j}$ is a family of atoms different from $*$.
    % \item $E_\bot^\top$ is the extensional completion of $|E|=\{\bot,\!\top\}$, $\le_E=id\mathrm{\cup}\{\bot\mathrm{<}\top\}$ and $j_E=\{(\emptyset,\!\top)\mathrm{\mapsto} \top\ ;\ (\{\top\},\bot)\mathrm{\mapsto}\bot\}$
    % \item $Norm^\top$ is the extensional completion of $D$ with $|E|=\{p,q,\top\}$, $\le_E=id\cup\{p<q<\top\}$ and $j_E=\{(\{p\},q)\mathrm{\mapsto} q , (\{q\},p)\mathrm{\mapsto} p , (\emptyset,\top)\mathrm{\mapsto}\top\}$
    % \item $Ex_1$ is the extensional completion of $|E|=\{*,p,u\}$, $\le_E=id$ and $j_E=\{(\emptyset,*)\mathrm{\mapsto} *\ ,\ (\{p\},u)\mathrm{\mapsto} u\ ,\ (\{*\},*)\mathrm{\mapsto} p\}$
  \end{enumerate}
\end{example}

\medskip

\noindent For the sake of simplicity, from now on we will work with a fixed extensional \Kweb $D$. Moreover, we will use the notation $a\cons\alpha:=i_D(a,\alpha)$ \newsyminvsec{reduction}{\protect\rta}. Notice that, due to the injectivity of $i_D$, any $\alpha\in D$ can be uniquely rewritten into $a\cons\alpha'$, and more generally into $a_1\cons\cdots\cons a_n\cons\alpha_n$ for any $n$. 

\begin{remark}
  Using these notations, the model $H^f$ can be summarized by writing, for each $n$:
  \[ \alpha_1^n\ =\ \underbrace{\emptyset\cons \cdots \cons \emptyset}_{f(n)}\cons \{\alpha_1^{n\+1}\}\cons * \]
\end{remark}

\subsubsection{Interpretation of the \Lcalcul} \quad \newline%
The Cartesian closed structure of \Scottlb endowed with the isomorphisms $app_{D}$ and $abs_{D}$ of the reflexive object induced by $D$ (see Proposition~\ref{prop:reflOb}) defines, in a standard manner, a model of the \Lcalcul.\todo{REF!}

A term $M$ with at most $n$ free variables $x_1,\dots,x_n$ is interpreted as the graph of a mor-\linebreak phism~$\llb M\rrb^{x_1...x_n}_{D}$ \newsyminvinvinv{interpretation}{\protect\llb .\protect\rrb^{\protect\vec x}_{D}}{\protect\Lamb in a K-model} from $D^n$ to $D$ (when $n$ is obvious, we can use \newsymsecinv{interpretation}{\protect\llb.\protect\rrb^{\protect\bar x}}{$\Lambda$ in a K-model}). By Equations \eqref{eq:graphf} and \eqref{eq:defRta} we have:
$$\llb M\rrb^{x_1...x_n}_{D}\subseteq\  (D\Rta\!\cdots\Rta D\Rta D)\ \simeq\ (\Achf{D}^{op})^n\times D. $$
In Figure~\ref{fig:intLam}, we explicit the interpretation $\llb M\rrb_{D}^{x_1...x_n}$ by structural induction on $M$.
\begin{figure*} \label{fig:TestInter}
  \caption{Direct interpretation of $\Lamb$ in $D$ \label{fig:intLam}}
    \begin{center}
      $\llb x_i \rrb_D^{\vec x}  =  \{(\vec a,\alpha)\ |\ \alpha\le\beta\in a_i\}$
      \hspace{6em}
      $ \llb \lambda y.M \rrb_D^{\vec x}  =  \{(\vec a,b\cons\alpha)\ |\ (\vec ab,\alpha)\in\llb M\rrb_D^{\vec xy}\}$
      \vspace{0.3em}\\
      $\!\llb M\ N \rrb_D^{\vec x}  =  \{(\vec a,\alpha)\ |\ \exists b,(\vec a,b\cons\alpha)\in\llb M\rrb_D^{\vec x}\ \wedge \forall \beta\mathrm{\in} b, (\vec a,\beta)\in\llb N\rrb_D^{\vec x}\}$
    \end{center}
\end{figure*}

\begin{example}
  \vspace{-0.5em}
  \begin{align*}
    \llb\lambda x.y\rrb_{D}^y &= \{((a), b\cons\alpha)\mid \alpha\le_{D}\beta\in a\},\\
    \llb\lambda x.x\rrb_{D}^y &= \{((a), b\cons\alpha)\mid \alpha\le_{D}\beta\in b\},\\
    \llb\I\rrb_{D} &= \{a\cons\alpha\mid \alpha\le_{D}\beta\in a\},\\
    \llb \underline{1}\rrb_{D} &= \{a\cons b\cons\alpha \mid \exists c, c\cons\alpha\le_{D}\beta\in a,\ c\le_{\Achf{D}}b\}.
  \end{align*}
  In the last two cases, terms are interpreted in an empty environment. We omit the empty sequence associated with the empty environment, {\em e.g.}, $a\cons b\cons\alpha$ stands for $((),a\cons b\cons\alpha)$.\\
  We can verify that extensionality holds, indeed $\llb \underline{1}\rrb_{D}=\llb\I\rrb_{D}$, since $c\cons\alpha\le_{D}\beta\in a$ and $c\le_{\Achf{D}}b$ exactly say that $b\cons\alpha\!\le_{D}\!\beta\!\in\! a$, and since any element of $\gamma\!\in\! D$ is equal to $d\cons\delta$ for a suitable $d$ and $\delta$.
\end{example}

\subsubsection{Intersection types} \quad \newline%
It is folklore that the interpretation of the \Lcalcul into a given \Kweb $D$ is characterized by a specific \newdefprem{intersection type system}{characterizing K-models}. In fact any element $\alpha\in D$ can be seen as an intersection type 
\begin{align*}
  \alpha_1\wedge\cdots\wedge\alpha_n&\rta\beta   &   \text{given by}\quad &\alpha= \{\alpha_1,\dots,\alpha_n\}\cons\beta.
\end{align*}
In Figure~\ref{fig:tyLam}, we give the intersection-type assignment corresponding to the \Kweb induced by $D$.%, {\em i.e.}, $\llb M\rrb_{D}^{x_1...x_n}$ is the set of $((\alpha_1,\dots,\alpha_n),\beta)$ such that 
%   $$x_1:\alpha_1,\dots,x_n:\alpha_n\vdash M:\beta\quad \text{ is derivable}. $$

\begin{figure*}
  \caption{Intersection type system computing the interpretation in $D$ \label{fig:tyLam}}
  \begin{center}
    \AxiomC{$\alpha\in a$}
    \RightLabel{{\scriptsize $(I\dash id)$}}
    \UnaryInfC{$x:a\vdash x:\alpha$}
    \DisplayProof\hskip 50pt
    \AxiomC{$\Gamma\vdash M:\alpha$}
    \RightLabel{{\scriptsize $(I\dash weak)$}}
    \UnaryInfC{$\Gamma,x:a\vdash M:\alpha$}
    \DisplayProof\hskip 50pt
    \AxiomC{$\Gamma\vdash M:\beta$}
    \AxiomC{$\alpha\le\beta$}
    \RightLabel{{\scriptsize $(I\dash \le)$}}
    \BinaryInfC{$\Gamma\vdash M:\alpha$}
    \DisplayProof\\ \vspace{1em}
    \AxiomC{$\Gamma,x:a\vdash M:\alpha$}
    \RightLabel{{\scriptsize $(I\dash \lambda)$}}
    \UnaryInfC{$\Gamma\vdash \lambda x.M:a\cons\alpha$}
    \DisplayProof\hskip 50pt
    \AxiomC{$\Gamma\vdash M:a\cons\alpha$}
    \AxiomC{$\forall \beta\in a,\ \Gamma\vdash N:\beta$}
    \RightLabel{{\scriptsize $(I\dash \at)$}}
    \BinaryInfC{$\Gamma\vdash M\ N:\alpha$}
    \DisplayProof
  \end{center}
\end{figure*}

\begin{prop}
  Let $M$ be a term of \Lamb, the following statements are equivalent:
  \begin{itemize}
  \item $(\vec a,\alpha)\in\llb M\rrb_D^{\vec x}$,
  \item the type judgment $\vec x:\vec a\vdash M:\alpha$ is derivable by the rules of Figure~\ref{fig:tyLam}.
  \end{itemize}
\end{prop}
\begin{proof}
By structural induction on the grammar of $\Lamb$.
\end{proof}

\subsection{The result} \quad \newline\label{ssec:results}%
% !TEX root = main.tex
% 
%
%
%The goal of this subsection is to present the full result (the proofs being developed in the next ones). But firstly we need to introduce, in Subsubsection~\ref{sssec:BasicDefBT}, a reminding about B\"ohm trees and their properties. Those definition will permit to present the last condition, the {\em commutation with B\"ohm trees} of Definition~\ref{def:comBTdef}, that we  add to our result. We will also present a simpler sufficient condition, the {\em weak positivity} of Definition~\ref{def:wpos}, whose main interest is the justification of our examples. Only then will we introduce the central property of this article, the {\em hyperimmunity} of Definition~\ref{def:hyperim}. Finally the last subsubsection present the main theorem (Theorem~\ref{th:final}) that characterize fully abstract extensional K-models that commutes with B\"ohm trees as being exactly the hyperimmune ones.
%
%
%
%
We state our main result, claiming an equivalence between hyperimmunity (Def.~\ref{def:hyperim}) and full abstraction for \Hst. %with two theorems (Th.~\ref{th:final} and \ref{th:final2}) that differ in their premises. We will prove later (Th.~\ref{th:AppTheAndSens}) that these premises are in indeed equivalent.

%Recall from Sections~\ref{ssec:K-models} and~\ref{ssec:preliminaries} that an extensional K-model is a pair~{$(D,(\_\rta\_))$} where:
%\begin{itemize}
%\item $D$ is a poset.
%\item $(\_\cons\_)$ is an order isomorphism between $D\Rta D$ and $D$ (where $D\Rta D\;\simeq\;\Achf D ^{op}\times D$).
%\end{itemize}
%In particular, for any $k\in\Nat$, an element $\alpha\in D$ can be unfolded into a unique sequence $\alpha&=a_{1}\cons\cdots\cons a_{k}\cons\alpha'$.

\begin{definition}[\newdef{Hyperimmunity}]\label{def:hyperim}
  A (possibly partial) extensional \Kweb $D$ is said to be \newdef{hyperimmune} if for every sequence  $(\alpha_n)_{n\ge 0}\in D^\Nat$, there is no recursive function $g:\Nat\mathrm{\rta}\Nat$ satisfying:
  \begin{align}
    &\forall n\ge 0, & \alpha_n&=a_{n,1}\cons\cdots\cons a_{n,g(n)}\cons\alpha_n' & \text{and} & & \alpha_{n\+1}&\in \bigcup_{k\le g(n)}a_{n,k}. \label{eq:hyperimmunity}
  \end{align}\smallskip
\end{definition}

\noindent Notice, in the above definition, that each antichain $a_{n,i}$ always exists and it is uniquely determined by the isomorphism between $D$ and $D\Rta D$ that allows us to unfold any element $\alpha_i$ as an arrow (of any length).

The idea is the following. The sequence $(\alpha_n)_{n\ge 0}$ is morally describing a non well-founded chain of elements of $D$, through the isomorphism $D\simeq D\Rta D$, allowing us to see any element $\alpha_i$ as an arrow (of any length):
%The idea of hyperimmunity is to forbid, for any recursive function $g$, a situation of the following shape:
\begin{alignat*}{4}
  \alpha_0= a_{0,1}\cons \cdots\; &  a_{0,i_0}\cdots \cons a_{0,g(0)}\cons \alpha_0' \hspace{-17pt}\\
&\ \rotatebox[origin=c]{90}{$\in$}\\
\vspace{-1em}
& \alpha_1= a_{1,1}\cons \cdots\; & a_{1,i_1}\cdots \cons a_{1,g(1)}\cons \alpha_1' \hspace{-21pt}\\
& & \omit\span\ \rotatebox[origin=c]{90}{$\in$}\\
& &\alpha_2=  a_{2,1}\cons \cdots &\ a_{2,i_2} \cdots \cons\; a_{2,g(2)}\cons \alpha_2'\\
& & & \quad \ddots
\end{alignat*}
The growth rate $(i_n)_n$ of the chain $(\alpha_n)_n$ depends on how many arrows must be displayed in $\alpha_i$ in order to see $\alpha_{i+1}$ as an element of the antecedent of one of them. Now, hyperimmunity means that if any such non-well founded chain $(\alpha_n)_n$ exists, then its growth rate $(i_n)_n$ cannot be bounded by any recursive function $g$.

\begin{remark}
  It would not be sufficient to simply consider the function $n\mapsto i_n$ such that $\alpha_{n+1}\mathrm{\in} a_{n,i_n}$ rather than the bounding function $g$. Indeed, $n\mapsto i_n$ may not be recursive even while $g$ is.
\end{remark}

%Seeing the model as an intersubsection type system, hyperimmunity somehow says that for B\"ohm trees that come for actual terms (and thus are computable), a co-inductive typing would reduce to an inductive one. This echos the requirement of commutation with B\"ohm trees that basicly says that the interpretation of a term corespond to the inductive typing of it B\"ohm tree.

\begin{prop} \label{prop:hypPartial}
For any extensional partial \Kweb $E$ (Def.~\ref{def:partialKweb}), the completion $\Compl{E}$ (Def.~\ref{def:Compl}) is hyperimmune iff $E$ is hyperimmune.
\end{prop}
\begin{proof}
  The left-to-right implication is trivial.\\
  The right-to-left one is obtained by contradiction:\\ 
  Assume to have a $(\alpha_n)_{n\ge 0}\in \Comp{E}^{\Nat}$ and a recursive function $g:\Nat\rta\Nat$ such that for all $n\ge 0$:
  \begin{align*}
    \alpha_n&=a_{n,1}\cons\cdots\cons a_{n,g(n)}\cons \alpha_n'  &  \text{and} & & \alpha_{n+1}\in\bigcup_{i\le g(n)} a_{n,i}
  \end{align*}
  Recall that the sequence $(E_k)_{k\ge 0}$ of Definition \ref{def:Compl} approximates the completion $\Comp{E}$.\\
  Then we have the following:
  \begin{itemize}
    \item There exists $k$ such that $\alpha_0\in E_k$, because $\alpha_0\in\Comp{E}=\bigcup_kE_k$.
    \item If $\alpha_n\in E_{j+1}$, then $\alpha_{n+1}\in E_{j}$, because there is $i\le g(n)$ such that $\alpha_{n+1}\in a_{n,i}\subseteq E_j$.
      %indeed there is $i\le g(n)$ such that $\alpha_{n+1}\in a_{ni}$ and an easy induction on $j$ gives that $a_{nj}\subseteq E_k$ for any $j$.
    \item If $\alpha_n\in E_0=E$, then $\alpha_{n+1}\in E$ by surjectivity of $j_E$.
    \end{itemize}
  Thus there is $k$ such that $(\alpha_n)_{n\ge k}\in E^\Nat$, which would break hyperimmunity of $E$.
\end{proof}

\begin{example}\label{ex:hyperimmunity}
  \begin{itemize}
  \item The well-stratified \Kwebs of Example~\ref{example:1}\eqref{enum:WellStrat} (and in particular $\Dinf$ of Item~\eqref{enum:Dinf}) are trivially hyperimmune: already in the partial K-model, there are not even $\alpha_1$, $\alpha_2$ and $n$ such that $\alpha_1=a_{1}\cons\cdots\cons a_{n}\cons\alpha_1'$ and $\alpha_2\in a_n$ (since $a_n=\emptyset$). The non-hyperimmunity of the partial \Kweb can be extended to the completion using Proposition~\ref{prop:hypPartial}.
  \item The model $\Compl{\omega}$ (Ex.~\ref{example:1}\eqref{enum:Ind}) is hyperimmune. Indeed, any such $(\alpha_n)_n$ in the partial K-model would respect $\alpha_{n+1}\mathrm{<_\Nat} \alpha_n$, hence $(\alpha_n)_n$ must be finite by well-foundedness of \Nat.
  \item The models \Pinf, $\Dinf^*$ and $\Compl{\mathbb{Z}}$ (Examples~\ref{example:1}\eqref{enum:Park}, \eqref{enum:D*} and \eqref{enum:CoInd}) are not hyperimmune. Indeed for all of them $g=(n\mapsto 1)$ satisfies the condition of Equation \eqref{eq:hyperimmunity}, the respective non-well founded chains $(\alpha_i)_i$ being $(*,*,\dots)$, $(p,q,p,q,\dots)$, and $(0,-1,-2,\dots)$:
      \begin{alignat*}9
        *=\ & \{*\} \rta *\hspace{-15pt} & & & 
        p=\ & \{q\} \rta p\hspace{-15pt} & & & 
        0=\ & \{1\} \rta 0 \hspace{-20pt}\\
        &\ \rotatebox[origin=c]{90}{$\in$} & & & 
        &\ \rotatebox[origin=c]{90}{$\in$} & & &
        &\ \rotatebox[origin=c]{90}{$\in$}\\
        &\ \ *=\  & \omit\span \{*\}\rta *  & 
        &\ q=\ & \omit\span \{p\}\rta q &
        &\ 1=\ &\omit\span \{2\}\rta 1\\
        & &\omit\span\ \ \rotatebox[origin=c]{90}{$\in$} &
        & &\omit\span\ \ \rotatebox[origin=c]{90}{$\in$} & 
        & &\omit\span\ \ \rotatebox[origin=c]{90}{$\in$}\\
        & &\ \ *=\ & \{*\}\rta * &
        & &\omit\span\ p= \{q\}\rta p &
        & &\omit\span\ 2= \{3\}\rta 2\\
        & & &\quad \ddots  & 
        & & \quad\quad\quad \ddots & &
        & & \quad\quad\quad \ddots
      \end{alignat*}
    \item More interestingly, the model $H^f$ (Ex.~\ref{example:1}\eqref{enum:Func}) is hyperimmune iff $f$ is a \emph{hyperimmune function} \cite{Nies}, {\em i.e.}, iff there is no recursive $g:\Nat\rta\Nat$ such that $f\le g$ (pointwise order); otherwise the corresponding sequence is $(\alpha_1^i)_i$.
    \begin{alignat*}4
      \alpha_1^0= \underbrace{\emptyset \rta \cdots \rta \emptyset}_{_{f(0)\ \text{times}}} \rta & \{\alpha_1^1\} \rta  \emptyset \rta \cdots \rta \emptyset\rta * \hspace{-20pt} \\ %could have use \mathop
      &\ \rotatebox[origin=c]{90}{$\in$}\\
      & \ \alpha_1^1= \underbrace{\emptyset \rta \cdots \rta \emptyset}_{_{f(1)\ \text{times}}} \rta & \omit\span  \{\alpha_1^2\} \rta  \emptyset \rta \cdots \rta \emptyset \rta *\\
      & & \rotatebox[origin=c]{90}{$\in$} \hspace{7.5em}\\
      & & \ \alpha_1^2= \underbrace{\emptyset \rta \cdots\rta \emptyset}_{_{f(2)\ \text{times}}} \rta & \{\alpha_1^3\} \rta  \emptyset \rta \cdots \rta \emptyset \rta  *\\
      & & &\ \ \rotatebox[origin=c]{90}{$\in$}\\
      & & &\quad  \ddots\\
    \end{alignat*}
  \end{itemize}
\end{example}

\noindent The following theorem constitutes the main result of the paper. It shows the equivalence between hyperimmunity and (inequational) full abstraction for \Hst under a certain condition. This condition, namely the approximation property, is a standard property that will be defined in more details in Definition~\ref{def:approximationProp}.

\begin{theorem}\label{th:final}
  For any extensional and approximable \Kweb $D$ (Def.~\ref{def:approximationProp}), the following are equivalent:
  \begin{enumerate}
  \item $D$ is hyperimmune, \label{eq:n1}
  \item $D$ is inequationally fully abstract for \Hst, \label{eq:n2}
  \item $D$ is fully abstract for \Hst. \label{eq:n3}
  \end{enumerate}
\end{theorem}

%\begin{theorem}{Main theorem (syntactic)}\label{th:final2}
%  For any extensional \Kweb $D$ sensible for \Lam{D} (Def.~\ref{def:sensibilityWTests}), the following are equivalent:
%  \begin{enumerate}
%  \item $D$ is hyperimmune, \label{eq:n1}
%  \item $D$ is inequationally fully abstract for \Hst, \label{eq:n2}
%  \item $D$ is fully abstract for \Hst. \label{eq:n3}
%  \end{enumerate}
%\end{theorem}

\smallskip

\begin{example}
  The model \Dinf (Ex.\ref{example:1}\eqref{enum:Dinf}), the model \Compl{\omega} (Ex.\ref{example:1}\eqref{enum:Ind}) and the well-stratified \Kwebs (Ex.\ref{example:1}\eqref{enum:WellStrat}) will be shown inequationally fully abstract, as well as the models~$H^f$ when~$f$ is hyperimmune. The models~$\Dinf^*$,~\Compl{\mathbb{Z}} (Ex.\ref{example:1}\eqref{enum:D*} and Ex.\ref{example:1}\eqref{enum:CoInd}) will not be, as well as the model~$H^f$ for~$f$ not hyperimmune.
  %Indeed, we have seen in Example~\ref{ex:hyperimmunity} that the first models are hyperimmune but not the second. %Add reference to a forthcoming article?
\end{example}

\section{Proof}
    \label{sec:proof1}
% !TEX root = main.tex
% 

The main idea of this proof is not new, it consists in using B\"ohm trees to decompose the interpretation of the \Lcalcul. In order to do so, we need to interpret them into our \Kweb $D$ so that the following diagram commutes:
\begin{center}
  \begin{tikzpicture}[description/.style={fill=white,inner sep=2pt},ampersand replacement=\&]
    \matrix (m) [matrix of math nodes, row sep=1em, column sep=5em, text height=1.5ex, text depth=0.25ex]
    { {\Lamb} \& \& D \\
      \& {\BT} \\};
    \path[->] (m-1-1) edge[] node[description] {$\llb.\rrb$} (m-1-3);
    \path[->] (m-1-1) edge[] node[description] {$\BT(.)$} (m-2-2);
    \path[->] (m-2-2) edge[] node[description] {$\llb.\rrb_*$} (m-1-3);
  \end{tikzpicture}
\end{center}
The approximation and quasi-approximation properties of Definitions~\ref{def:approximationProp} and~\ref{def:quasi-approxTh} exactly state this decomposition for two specific choices of interpretation. Indeed, we will see in Definition~\ref{def:BTCohInterpret} that there are many different possible interpretations of the B\"ohm trees, we will mainly focus on the inductive interpretation (Def.~\ref{def:(co)indInt}) and the quasi-finite interpretation (Def.~\ref{def:BTinterpretQF}).

The approximation and quasi-approximation properties will have different roles. The approximation property, {\em i.e.}, the decomposition via the inductive interpretation, mainly says that the interpretation of terms is approximable by finite B\"ohm trees. Approximation property is a hypothesis of Theorem~\ref{th:final} and it holds in all known candidates to full abstraction, {\em i.e.}, extensional and sensible models (Ex.~\ref{ex:ApprThe}). We even conjecture, in fact, that all K-models that are fully abstract for \Hst respect the approximation property.

The quasi-approximation property is a fairly finer property\footnote{Even if technically independent.} that is based on deep references to recursivity theory. The quasi-approximation property will be proved equivalent to both full abstraction for \Hst and hyperimmunity in the presence of the approximation property.

\begin{theorem}%{Developed semantic theorem}
  For any extensional and approximable \Kweb $D$, the following are equivalent:
  \begin{enumerate}
  \item $D$ is hyperimmune, \label{eq:qa0}
  \item $D$ respects the quasi-approximation property, \label{eq:qa1}
  \item $D$ is inequationally fully abstract for \Lamb, \label{eq:qa2}
  \item $D$ is fully abstract for \Lamb. \label{eq:qa3}
  \end{enumerate}
\end{theorem}
\begin{proof} \quad
  \begin{itemize}
  \item $\eqref{eq:qa0}\Rta\eqref{eq:qa1}$: Theorem~\ref{th:Hy+Ap=QAp},
  \item $\eqref{eq:qa1}\Rta\eqref{eq:qa2}$: inequational adequacy is the object of Theorem~\ref{th:QApp->Adeq} and inequational completeness the one of Theorem~\ref{th:QApp->FullCompl},
  \item $\eqref{eq:qa2}\Rta\eqref{eq:qa3}$: trivial,
  \item $\eqref{eq:qa3}\Rta\eqref{eq:qa0}$: Theorem~\ref{th:countex}.\qedhere  
  \end{itemize}
\end{proof}

  \subsection{B\"ohm trees}
    \label{ssec:BohmTrees}
% !TEX root = main.tex
%

\subsubsection{Basic definitions}  \quad \newline%
The B\"ohm trees provide one of the simplest semantics for the \Lcalcul:

\begin{definition}\label{def:BT}
  The set of \newdef{B\"ohm trees} is the co-inductive structure generated by the grammar:
    \begin{center}
      \begin{tabular}{l l r @{\ ::=\ } l l}
        (B\"ohm trees) & 
          \newsym{\protect\BT} \hspace{2em}&
          $U,V$ \hspace{1em} & 
          \hspace{2em} $\newsym{\protect\Omega} \quad |\quad \lambda x_1...x_n.y\ U_1\cdots U_k$ & $,\forall n,\forall k\ge 0$
      \end{tabular}
    \end{center}
  The \newdefpremsec{B\"ohm tree}{of a $\lambda$-term} $M$ ({\em i.e.}, its interpretation), is defined by co-induction:
  \begin{itemize}
  \item If $M$ head diverges, then $\BT(M)=\Omega$,
  \item if $M\rta_h^*\lambda x_1...x_n.y\ N_1\cdots N_k$, then
    $$\BT(M)=\lambda x_1...x_n.y\ \BT(N_1)\cdots \BT(N_k).$$
  \end{itemize}
  Notice that a B\"ohm tree can be described as a finitely branching tree (of possibly infinite height) where nodes are labeled either by a constant $\Omega$, or by a list of abstractions and by a head variable.

  Capital final Latin letters $U,V,W...$ will range over B\"ohm trees.
\end{definition}

\begin{example}
  The B\"ohm trees $\BT(\lambda x. x\ (\lambda y.x\ y))$, $\BT(x\ (\I\ \I)\ (y\ (\FixP\ \I)))$, \FixP and $\BT(\Theta\ (\lambda uxy. y (u\ x))\ z)$ are described in Figure~\ref{fig:exBT}.
\end{example}
\begin{figure*}
  \caption{Some examples of B\"ohm trees.}\label{fig:exBT}
    \begin{tikzpicture}%[description/.style={fill=white,inner sep=2pt},ampersand replacement=\&]
        \node (nom1) at (0,3) {$\BT(\lambda x. x\ (\lambda y.x\ y))$:};
        \node (a) at (0,2) {$\lambda x.x\ .\!\!\!$};
        \node (b) at (0.47,1) {$\lambda y.x\ .\!\!\!$};  
        \node (c) at (0.94,0) {$y$};
        \draw[-] (a.south east) -- node [auto] {} (b.north);
        \draw[-] (b.south east) -- node [auto] {} (c.north);
        \node (nom2) at (4,3) {$\BT(x\ (\I\ \I)\ (y\ (\FixP\ \I)))$:};
        \node (a) at (4,2) {$ x\ .\!\!\!$};
        \node (a') at (4.4,2) {$\vphantom{x}.\!\!\!$};
        \node (b) at (4.8,1) {$ y\ .\!\!\!$};  
        \node (c) at (5.03,0) {$\Omega$};
        \node (d) at (3.8,1) {$\lambda x.x$};
        \draw[-] (a'.south east) -- node [auto] {} (b.north);
        \draw[-] (b.south east) -- node [auto] {} (c.north);
        \draw[-] (a.south east) -- node [auto] {} (d.north);
        \node (nom3) at (7.5,3) {$\FixP$:};
        \node (a) at (7.5,2) {$ \lambda f.f\ .\!\!\!$};
        \node (b) at (7.99,1) {$ f\ .\!\!\!$};  
        \node (c) at (8.25,0) {$f\ .\!\!\!$};
        \node (d) at (8.5,-1) {$\vdots$};
        \draw[-] (a.south east) -- node [auto] {} (b.north);
        \draw[-] (b.south east) -- node [auto] {} (c.north);
        \draw[-] (c.south east) -- node [auto] {} (d.north);
        \node (nom3) at (11,3) {$\BT(\Theta\ (\lambda uxy.y\ (u\ x))\ z)$:};
        \node (a) at (11,2) {$ \lambda y_1.y_1\ .\!\!\!$};
        \node (b) at (11.62,1) {$ \lambda y_2.y_2\ .\!\!\!$};  
        \node (c) at (12.22,0) {$\lambda y_3.y_3\ .\!\!\!$};
        \node (d) at (13.18,-1) {$\ddots$};
        \draw[-] (a.south east) -- node [auto] {} (b.north);
        \draw[-] (b.south east) -- node [auto] {} (c.north);
        \draw[-] (c.south east) -- node [auto] {} (d.north west);
    \end{tikzpicture}
\end{figure*}

There exist B\"ohm trees that do not come from terms: %, in their all generality there are some that are not B\"ohm tree of any term. 

\begin{example}
  A B\"ohm tree with infinitely many free variables (such as the first one below) cannot be obtained from $\lambda$-terms that have finitely many free variables. Worse, if $g:\Nat\rta\Nat$ is non recursive, then the second B\"ohm tree below does not come from any term (otherwise it would be possible to compute $g$ from this term).
    \begin{alignat*}8
      x_0\ &. & &  \hspace{5em}  &  \lambda x_1. x_0\ x_0\mathop{\cdots}_{g(0)}x_0\ &.\\
      & | & &                    &   & |\\
      & x_1\ & .\hspace{1em} &   &   & \lambda x_2.x_1\ x_1\mathop{\cdots}_{g(1)}x_1\ & . \\
      & & |\hspace{1em} &        &   & & \omit\span |\\
      & & x_2\ &.                &   & & \omit\span\lambda x_3.x_2\ x_2\mathop{\cdots}_{g(2)}\ x_2\ &. \\
      & & & |                    &   &  & & \omit\span|\\
      & & & \vdots               &   & & & \omit\span \ddots \hspace{-1em}
    \end{alignat*}
\end{example}

%As for terms, the B\"ohm trees can be described equivalently with De Bruijn.
%
%\begin{definition}
%  There is a one-to-one correspondance between B\"ohm trees and DeBruijn-B\"ohm trees defined as the co-inductive structure:
%  \begin{center}
%    \begin{tabular}{l r @{\ ::=\ } l l}
%      (DeBruijn-B\"ohm trees) \hspace{2em} & 
%      $U,V$ \hspace{1em} & 
%      \hspace{2em} $\protect\Omega \quad |\quad \lambda^n.i\ U_1\cdots U_k$ & $,\forall k,n,i\in\Nat$
%    \end{tabular}
%  \end{center}
%  where the $i$ is used to denotes the $(n-i)^{th}$ abstraction if $i<n$ or the $(i-n)^{th}$ variable of the the father of your term (if it is not free).
%
%  We will define a other (nearly) equivalent presentation, that we call BiDeBruijn-B\"ohm trees, as the co-inductive structure:
%  \begin{center}
%    \begin{tabular}{l r @{\ ::=\ } l l}
%      (BiDeBruijn-B\"ohm trees) \hspace{2em} & 
%      $U,V$ \hspace{1em} & 
%      \hspace{2em} $\protect\Omega \quad |\quad \lambda^n.(h,j)\ U_1\cdots U_k$ & $,\forall k,n,h,i\in\Nat$
%    \end{tabular}
%  where the height indice $h$ and the abstraction indice $i$ are used to denotes the $i^{th}$ abstraction for the $h^{th}$ father of the substree. A BiDeBruijn-B\"ohm three is well formed if for every subtree $\lambda^n.(h,j)\ U_1\cdots U_k$, either the subtree has no $h^{th}$ father of this $h^{th}$ has mor the $i$ abstractions. There is one-to-one corespondance between B\"ohm threes and BiDeBruijn-B\"ohm three.
%  \end{center}
%\end{definition}

\subsubsection{Properties} \quad \newline%
The B\"ohm tree model model carries several interesting properties for the study of the untyped \Lcalcul. 
%First, it is an actual model. \todo{not necessarily true if we change the notion of model.}
%
%\begin{definition}
%  We denote $\equiv_{\BT}$ the equivalence induced by B\"ohm trees on \Lamb and defined by $M\equiv_{\BT} N$ iff $\BT(M)=\BT(N)$.
%\end{definition}
%
%\begin{prop}[\cite{Barendregt}, Proposition 16.4.2 (i)]\label{prop:BTinv}
%  The relation $\equiv_{\BT}$ is a congruence.
%\end{prop}
%
%\begin{prop}[\cite{Barendregt}, Proposition 10.1.6]\label{prop:BTinv}
%  The relation $\equiv_{\BT}$ is a \Ltheory, {\em i.e.}, for every $M,N\in\Lambda$ such that $M\rta N$, $\BT(M)=\BT(N)$.
%\end{prop}
%
%But it goes further. 
By construction, it is sensible for the head reduction, and, moreover, it is adequate for~\Hst which is coarser. \todo{prop}%  
%Not only is it a sensible model (Def.~\ref{def:sensibleModel}) for the head reduction (trivial by construction), but there are adequation theorems (Def.~\ref{def:AdeqTh}) for \ThNF and, {\em a fortiori}, \Hst (Ex.~\ref{ex:H*}). 
Moreover, those properties extend to inequations using the following natural notion of inclusion on B\"ohm trees:

\begin{definition}
  The \newdefpremsec{inclusion}{of B\"ohm trees} $U\newsymprem{\protect\subseteq}{for B\"ohm trees} V$ %(morally $U$ is a prefix of $V$) 
  is co-inductively defined by:
  \begin{itemize}
  \item $\Omega\subseteq V$ for all $V$
  \item If for all $i\le k$, $U_i\subseteq V_i$, then 
    $$(\lambda x_1...x_n.y\ U_1\cdots U_k)\ \subseteq (\lambda x_1...x_n.y\ V_1\cdots V_k).$$
  \end{itemize}
  For readability, we will write $M\subseteq_{\BT} N$ whenever $\BT(M)\subseteq \BT(N)$.
\end{definition}

The lower bounds of a B\"ohm tree $U$ are obtained by replacing (possibly infinitely many) subtrees of $U$ by $\Omega$.

\begin{example}
  For any $M$, we have the inclusion 
  $$\FixP\ (\lambda uxy. x\ (u\ y)\ \Om)\quad \subseteq_{\BT}\quad \FixP\ (\lambda uxy.x\ (u\ y)\ (M\ x))$$
  \hspace{11.5em}
  \begin{tikzpicture}[description/.style={fill=white,inner sep=2pt},ampersand replacement=\&]
    \matrix (m) [matrix of math nodes, row sep=0.5em, column sep=0.7em, text height=1.5ex, text depth=0.25ex]
    { \lambda x_0x_1. x_0\ .\ . \& \& \&
      \lambda x_0x_1.x_0\ .\ . \\
      \lambda x_2.x_1\ .\ . \& \Omega  \&  \subseteq \&
      \lambda x_2.x_1\ .\ . \& {\BT(M\ x_1)} \\ 
      \lambda x_3.x_2\ .\ . \& { \Omega} \& \&
      \lambda x_3.x_2\ .\ . \& { \BT(M\ x_2)} \\
      \vdots \& \Omega \& \&
      \quad\vdots\!\!\ddots \& { \BT(M\ x_3)}\quad \\};
    \path[-] (m-1-1) edge[] node[auto] {} (m-2-1);
    \path[-] (m-1-1) edge[] node[auto] {} (m-2-2);
    \path[-] (m-2-1) edge[] node[auto] {} (m-3-1);
    \path[-] (m-2-1) edge[] node[auto] {} (m-3-2);
    \path[-] (m-3-1) edge[] node[auto] {} (m-4-1);
    \path[-] (m-3-1) edge[] node[auto] {} (m-4-2);
    \path[-] (m-1-4) edge[] node[auto] {} (m-2-4);
    \path[-] (m-1-4) edge[] node[auto] {} (m-2-5);
    \path[-] (m-2-4) edge[] node[auto] {} (m-3-4);
    \path[-] (m-2-4) edge[] node[auto] {} (m-3-5);
    \path[-] (m-3-4) edge[] node[auto] {} (m-4-4);
    \path[-] (m-3-4) edge[] node[auto] {} (m-4-5);
  \end{tikzpicture}
\end{example}
\begin{prop}[{\cite[Proposition~16.4.7]{Barendregt}}]
  B\"ohm trees are inequationally adequate for \Hst, {\em i.e.}
\[
  \text{if }\quad M\subseteq_{\BT}N\quad\text{ then }\quad M\leob N
  \]
\end{prop}\medskip

\noindent The converse does not hold (because $\subseteq_{\BT}$ is not extensional), so that we do not have full abstraction, but rather a new (inequational) $\lambda$-theory called \newsym{\protect\mathcal{BT}}.

Forcefully adding the extensionality in~$\mathcal{BT}$, we obtain the theory~$\mathcal{BT}\!\eta$ which is different from~\Hst:

\begin{example}\label{ex:J}
  The term $\J=\Theta\ (\lambda uxy. x\ (u\ y))$ defines the following B\"ohm tree:
  \begin{alignat*}4
    \lambda x_0x_1. x_0\ &.\\
    & |\\
    & \lambda x_2.x_1\  &\omit\span . \\
    & & \omit\span |\\
    & & \lambda x_3.x_2\  &. \\
    & & & \ddots 
  \end{alignat*} 
  The behavior of this term is the same as the identity, so that we have $\J\equivob \I$, but their B\"ohm trees are distinct and they are not $\eta$-convertible, so that $\J\not\equiv_{\mathcal{BT}\!\eta}\I$.
\end{example}

\subsubsection{B\"ohm trees and full abstraction}\label{sssec:BT+FA} \quad \newline%
We have seen that \BT is not fully abstract for \Hst since it is not extensional; however, there are refinements using the notion of infinite $\eta$ expansion that permit to say something about the full abstraction (Proposition~\ref{prop:H*le}).

\begin{definition}
   We write by \newsym{\protect\succeq_{\eta}} the \newdefpremsec{$\eta$-reduction}{on B\"ohm trees}, that is $U\succeq_{\eta}V$ if $U=V=\Omega$ or if
   \begin{alignat*}6
     &&  U&=\lambda x_1...x_{n+m}.&&y\ V_1\cdots V_k\ x_{n+1}\cdots x_{n+m}\\
     &\text{and } & V&=\lambda x_1...x_{n}.&&y\ V_1\cdots V_k
   \end{alignat*}
   where $x_{n+1},....,x_{n+m}\not\in \FV(V_1,....,V_k$).
\end{definition}
   
\begin{definition}\label{def:etainf}
   We write by \newsym{\protect\geetinf} the co-inductive version of $\succeq_\eta$, that is the coinductive relation generated by:
   \begin{center}
    \AxiomC{$\vphantom{A}$}
    \RightLabel{\footnotesize  \newsymsec{rule}{(\protect\mathtt{\eta\protect\infty\omega})}}
    \UnaryInfC{$\Omega\succeq_\eta\Omega$}
    \DisplayProof \hspace{20pt}
    \AxiomC{$\forall i\le k,\ U_i\geetinf V_i$}
    \AxiomC{$\forall i\le m,\ U_{k+i}\geetinf x_{n+i}$}
    \RightLabel{\footnotesize  \newsymsec{rule}{(\protect\mathtt{\eta\protect\infty\protect\at})}}
    \BinaryInfC{$\lambda x_1...x_{n+m}.y\ U_1\cdots U_{k+m}\geetinf\lambda x_1...x_n.y\ V_1\cdots V_k$}
    \DisplayProof 
  \end{center}
  By abuse of notations, given two $\lambda$-terms $M$ and $N$, we say that $M$ \newdef{infinitely $\eta$-expands} $N$, written~$M\geetinf N$, if~$\BT(M)\geetinf\BT(N)$.
\end{definition}

\begin{example}
  We have the inequations:
  $$\hspace{2.5em} \BT(\I)\quad\ \leetinf\quad \BT(\J) \quad \leetinf\quad\ \BT(\FixP\ (\lambda uxyz. x\ (u\ y)\ (u\ z)))$$
  \hspace{9em}
  \begin{tikzpicture}[description/.style={fill=white,inner sep=2pt},ampersand replacement=\&]
    \matrix (m) [matrix of math nodes, row sep=0.5em, column sep=0em, text height=1.5ex, text depth=0.25ex]
    { \lambda x_0.x_0 \& \&
      \lambda x_0x_1. x_0\ . \& \&
      \lambda x_0x_1y_1.x_0\ .\ . \\
      \& {  \preceq_{\eta\infty}} \&
      \lambda x_2.x_1\ . \& { \preceq_{\eta\infty}} \&
      \lambda x_2y_2.x_1\ .\ . \& \lambda y_2z_2. y_1\ .\ . \\ 
      \& \&
      \lambda x_3.x_2\ . \& \&
      \lambda x_3y_3.x_2\ .\ . \& \lambda y_3z_3. y_2\ .\ . \& \vdots\!\!\ddots\quad \\
      \& \&
      \vdots  \& \&
      \quad\vdots\!\!\ddots \& \vdots\!\!\ddots\quad\quad \& \vdots\!\!\ddots\quad \\};
    \path[-] (m-1-3) edge[] node[auto] {} (m-2-3);
    \path[-] (m-2-3) edge[] node[auto] {} (m-3-3);
    \path[-] (m-3-3) edge[] node[auto] {} (m-4-3);
    \path[-] (m-1-5) edge[] node[auto] {} (m-2-5);
    \path[-] (m-1-5) edge[] node[auto] {} (m-2-6);
    \path[-] (m-2-5) edge[] node[auto] {} (m-3-5);
    \path[-] (m-2-5) edge[] node[auto] {} (m-3-6);
    \path[-] (m-2-6) edge[double] node[auto] {} (m-3-7);
    \path[-] (m-3-5) edge[] node[auto] {} (m-4-5);
    \path[-] (m-3-5) edge[] node[auto] {} (m-4-6);
    \path[-] (m-3-6) edge[double] node[auto] {} (m-4-7);
  \end{tikzpicture}
\end{example}

\begin{remark}
  The $\eta$-reduction on B\"ohm trees is not directly related to the $\eta$-reduction on $\lambda$-terms.

  For example 
  $$\FixP\ (\lambda uzx. x\ (y\ z)) \not\leet \lambda x.\FixP\ (\lambda uzx. x\ (y z))\ x.$$
  Since $x$ is not free, however, this reduction holds at the level of B\"ohm trees.
  
  Conversely, we have 
  $$\FixP\ (\lambda uz. z\ (u\ z))\leet \FixP\ (\lambda uzx. z\ (u\ z)\ x)$$
  even while the B\"ohm trees are fairly different.

  However, the $\eta$-reduction on $\lambda$-terms is directly implied by the infinite $\eta$ reduction.
\end{remark}

Using $\leetinf$, we can characterize the notion of observational equivalence ({\em i.e.}, \Hst)

\begin{prop}[{\cite[Theorem~19.2.9]{Barendregt}}]\label{prop:H*le}
  For any terms $M,N\in\Lamb$, $M\leob N$ iff there exist two B\"ohm trees $U,V$ such that:
    $$\BT(M) \preceq_{\eta_\infty}U\subseteq V\geetinf \BT(N).$$
  %The symmetric closure of $\geetinf$ is the observational equivalence in $\Lamb$. 
\end{prop}

\begin{example} 
  In \Hst, we have the equivalence:
  $$\J\hspace{3em} \equivob\hspace{3em} \FixP\ (\lambda uxyz. x\ y\ (u\ z))$$
  \hspace{8em}
  \begin{tikzpicture}[description/.style={fill=white,inner sep=2pt},ampersand replacement=\&]
    \matrix (m) [matrix of math nodes, row sep=0.8em, column sep=0em, text height=1.5ex, text depth=0.25ex]
    { \lambda x_0x_1. x_0\ . \& \&
      \lambda x_0x_1y_1.x_0\ .\ . \hspace{-0.5em} \& \& \&
      \lambda x_0x_1y_1.x_0\ x_1\ . \\
      \lambda x_2.x_1\ . \& { \preceq_{\eta\infty}} \&
      \lambda x_2.x_1\ . \&  \hspace{-0.5em}\lambda x_2y_2.y_1\ x_2\ . \& { \succeq_{\eta\infty}} \&
      \lambda x_2y_2.y_1\ x_2\ . \\ 
      \lambda x_3.x_2\ . \& \&
      \lambda x_3.x_2\ . \&  \hspace{-0.5em}\lambda x_3y_3.y_2\ x_3 \& \&
      \lambda x_3y_3.y_2\ x_3\ . \\
      \vdots \& \&
      \vdots  \& \vdots \& \&
      \vdots\\
    };
    \path[-] (m-1-1) edge[] node[auto] {} (m-2-1);
    \path[-] (m-2-1) edge[] node[auto] {} (m-3-1);
    \path[-] (m-3-1) edge[] node[auto] {} (m-4-1);
    \path[-] (m-1-3) edge[] node[auto] {} (m-2-3);
    \path[-] (m-1-3) edge[] node[auto] {} (m-2-4);
    \path[-] (m-2-3) edge[] node[auto] {} (m-3-3);
    \path[-] (m-2-4) edge[] node[auto] {} (m-3-4);
    \path[-] (m-3-3) edge[] node[auto] {} (m-4-3);
    \path[-] (m-3-4) edge[] node[auto] {} (m-4-4);
    \path[-] (m-1-6) edge[] node[auto] {} (m-2-6);
    \path[-] (m-2-6) edge[] node[auto] {} (m-3-6);
    \path[-] (m-3-6) edge[] node[auto] {} (m-4-6);
  \end{tikzpicture}
\end{example}

\noindent The following trivial corollary will be rather useful for proving observational equivalences:

\begin{corollary}\label{cor:geetinfImpEq}
  For all $M,N\in\Lamb$,
  $$ M \geetinf N\ \Rta\ M\equivob\ N.$$
\end{corollary}
\begin{proof}
  By Proposition~\ref{prop:H*le} and since  $\BT(M)\leetinf \BT(M)\subseteq\BT(M)\geetinf\BT(N).$
\end{proof}

% !TEX root = main.tex
% 

\subsubsection{Subclasses of B\"ohm trees} \quad \newline%
Before saying anything on interpretation of B\"ohm trees in a \Kweb, we define some subclasses of B\"ohm trees that will work as potential bases. Such bases can be used to interpret a B\"ohm tree in our models as the sup of the interpretations of its approximants.\footnote{We will see that as a coinductive structure, a B\"ohm trees may have several possible interpretations into a given model.}

The only base that appears in the literature is the class \BTf of finite B\"ohm trees. However, we will oppose it the larger classes \BTomf and \BTqf of $\Omega$-finite and quasi-finite B\"ohm trees. %Larger bases are generally less interesting, but here we will see that they enforce a notion of ``stability'' for recursive B\"ohm trees.
The~$\Omega$-finiteness when applied to an approximant of an actual term (via its translation into a B\"ohm tree) is a property that insure the recursivity of the tree (Lemma.~\ref{lm:omfBT}). The quasi-finite B\"ohm trees are the~$\Omega$-finite B\"ohm trees that are somehow ``stable'' with respect to $\leetinf$ and $\geetinf$ (Lemma.~\ref{lemma:distrEtinfSubset}).

\begin{definition}\label{def:f&omfBT}
  We define the following classes over B\"ohm trees:
  \begin{itemize}
  \item The set of \newdefsecprem{B\"ohm tree}{finite}s, denoted \newsym{\protect\BTf}, is the set of B\"ohm trees inductively generated by the grammar of Definition~\ref{def:BT} (or equivalently B\"ohm trees of finite height). Given a term~$M$, we denote~$\BTf(M)$ the set of finite B\"ohm trees $U$ such that $U\subseteq\BT(M)$.
  \item The set of \newdefsecprem{B\"ohm tree}{$\Omega$-finite}s, denoted \newsym{\protect\BTomf}, is the set of B\"ohm trees that contain a finite number of occurrences of $\Omega$. %Given a term $M$, we denote $\BTomf(M)$ the set of $\Omega$-finite B\"ohm trees $U$ such that $U\subseteq\BT(M)$.
%  \item The set of \newdefsecprem{B\"ohm tree}{recursively sighted}s, denoted \newsym{\protect\BTrb}, is the set of B\"ohm trees where ocurences of each (free and bounded) variables are recursivelly bounded. Formally, there is a recursive function $g$ such that variables introduced at depth $n$ cannot be used at depth greater that $g(n)$. 
  \item The set of \newdefsecprem{B\"ohm tree}{quasi-finite}, denoted \newsym{\protect\BTqf}, is the set of those $\Omega$-finite B\"ohm trees having their number of occurences of each (free and bounded) variables recursively bounded. Formally, there is a recursive function $g$ such that variables abstracted at depth\footnotemark $n$ cannot occur at depth greater than $g(n)$.
  \end{itemize}
  Capital final Latin letters $X,Y,Z...$ will range over any of those classes of B\"ohm trees. We will use the notation \newsym{\protect\subseteq_f} (resp. \newsym{\protect\subseteq_{\Omega f}} and \newsym{\protect\subseteq_{qf}}) for the inclusion restricted to $\BTf\times\BT$ (resp.~$\BTomf\times\BT$ and~$\BTqf\times\BT$).
\end{definition}\footnotetext{We consider that free variables are ``abstracted'' at depth $0$.}

In particular, to any finite B\"ohm tree $U$ corresponds a term $M$ obtained by replacing every symbol $\Omega$ by the diverging term $\Om$. By abuse of notation, we may use one instead of the other.

\begin{example}
  The identity $\I$ corresponds to a finite B\"ohm tree and thus is in all three classes. The term $\lambda z.\FixP\ (\lambda ux.z\ u)$ has a B\"ohm tree that is $\Omega$-finite but not quasi-finite. The term $\FixP\ (\lambda ux. x\ u\ \Omega)$ has a B\"ohm tree that is neither of these classes. 
    \begin{align*}
      \BT(\lambda z.\FixP\ (\lambda ux.z\ u)) = \lambda zx_1.z\ & .             &      \BT(\FixP\ (\lambda ux. x\ u\ \Omega)) = \lambda x_1. x_1\ & .\ \Omega\\
                      & |             &                          & |  \\
      \lambda x_2.z\  & .             &         \lambda x_2 .x_2\ & .\ \Omega \\
                      & |             &                          & | \\
      \lambda x_3. z\ & .             &         \lambda x_3.x_3\  & .\ \Omega \\
                      & \vdots        &                           & \vdots \\
    \end{align*}
\end{example}

%The interest of these classes are of three orders. First is their density in \BT, in the sens that for each of this classes $\BT_X$, any B\"ohm tree $U\in\BT$ is the sup of the trees that are bellow but in this class $U=\bigcup\{V\in\BT_X\mid V\subseteq U\}$. The second interest is their different recursive properties, in particular we will see that any $\Omega$-finite B\"ohm trees bellow the B\"ohm tree of a term  $U\subseteq \BT(M)$ is a recursive tree (lemma~\ref{lm:omfBT}). The last interesting property is their behavior toward the $\eta\infty$-reduction/expansion (Lm.s~\ref{lemma:distrEtinfSubset} and~\ref{lemma:distrEtinfSubsetQF}).

\begin{lemma}\label{lm:omfBT}
  For all terms $M$, if $X\in\BTomf$ and $X\subseteq\BT(M)$, then $X$ is a recursive B\"ohm tree.
\end{lemma}
\begin{proof}
  First remark that only $X$ has to be recursive, not the proof of $X\subseteq\BT(M)$. Moreover, we only have to show that there exists a recursive construction of $X$, we do not have to generate it constructively.

   There is a finite number of $\Omega$'s in $X$ whose positions $p\in P$ can be guessed beforehand by an oracle that is finite thus recursive. After that, it suffices to compute the B\"ohm tree of $M$ except in these positions where we directly put an $\Omega$. This way the program is always productive as any $\Omega$ of $M$ ({\em i.e.}, any non terminating part of the process of computation of $\BT(M)$) will be shaded by a guessed $\Omega$ of $X$ (potentially far above).
\end{proof}

\begin{lemma}\label{lemma:OmegaBijEtInf}
  Let $U,V\in\BT$. If $U\leetinf V$ (def.~\ref{def:etainf}), there is a bijection between the $\Omega$'s in $U$ and those in $V$.
\end{lemma}
\begin{proof}
  Recall that $U\leetinf V$ is the relation whose proofs range over the coinductive sequents generated by
   \begin{center}
    \AxiomC{$\vphantom{A}$}
    \RightLabel{\footnotesize  $(\mathtt{\eta\infty\omega})$}
    \UnaryInfC{$\Omega\succeq_\eta\Omega$}
    \DisplayProof \hspace{20pt}
    \AxiomC{$\forall i\le k,\ U_i\geetinf V_i$}
    \AxiomC{$\forall i\le m,\ U_{k+i}\geetinf x_{n+i}$}
    \RightLabel{\footnotesize  $(\mathtt{\eta\infty\at})$}
    \BinaryInfC{$\lambda x_1...x_{n+m}.y\ U_1\cdots U_{k+m}\geetinf\lambda x_1...x_n.y\ V_1\cdots V_k$}
    \DisplayProof 
  \end{center} 
  Remark that this system is deterministic so that a sequent $U\leetinf V$ has at most one proof. In particular the occurrences of rule $(\mathtt{\eta\infty\omega})$ describe the pursued bijection.
\end{proof}

\begin{lemma}\label{lm:invarianceEtinfRB}
  For all $U,V\in \BT$ such that $U\leetinf V$,  $U\in\BTqf$ iff $V\in\BTqf$.
\end{lemma}
\begin{proof}
  By Lemma~\ref{lemma:OmegaBijEtInf}, we know that $U\in\BTomf$ iff $V\in\BTomf$.

  It is easy to see that if variable occurrences are bounded by $g$ in $U$, then they will be bounded by~$(n\mapsto max(g(n),1))$ in $V$ and conversly. Indeed an $\eta\infty$-expansion/reduction will not change the depth of any variable, and will only delete/introduce abstraction whose variable will be used exactly once at depth $1$.
\end{proof}

\begin{lemma}\label{lemma:distrEtinfSubset}
  Both ordering $\leetinf$ and $\geetinf$ distribute over $\subseteq_{qf}$, and the ordering~$\geetinf$ distributes over~$\subseteq_f$:
  \begin{itemize}
  \item For all $U,V\in\BT$ and $X\in\BTqf$ such that $X\subseteq_{qf} U\leetinf V$, there is $Y\in\BTqf$ such that\footnote{This is a commuting diagram, the $\rightsquigarrow$ arrow only recalls that $Y$ is obtained from $X$, $U$ and $V$.}
  \begin{alignat*}6
    U\ \ \ & \leetinf & V\ \ \ \\
    \rotatebox[origin=c]{90}{$\subseteq$}_{qf} & \ \rotatebox[origin=c]{-45}{$\rightsquigarrow$} & \rotatebox[origin=c]{90}{$\subseteq$}_{qf} \\
    X\ \ \ & \leetinf & \ Y.\ \ \ 
  \end{alignat*}
  \item For all $U,V\in\BT$ and $X\in\BTqf$ such that $X\subseteq_{qf} U\geetinf V$, there is $Y\in\BTqf$ such that
  \begin{alignat*}6
    U\ \ \ & \geetinf & V\ \ \ \\
    \rotatebox[origin=c]{90}{$\subseteq$}_{qf} & \ \rotatebox[origin=c]{-45}{$\rightsquigarrow$} & \rotatebox[origin=c]{90}{$\subseteq$}_{qf} \\
    X\ \ \ & \geetinf & \ Y.\ \ \ 
  \end{alignat*}
  \item For all $U,V\in\BT$ and $X\in\BTf$ such that $X\subseteq_f U\geetinf V$, there is $Y\in\BTf$ such that 
  \begin{alignat*}6
    U\ \ \ & \geetinf & V\ \ \ \\
    \rotatebox[origin=c]{90}{$\subseteq$}_{f\;} & \ \rotatebox[origin=c]{-45}{$\rightsquigarrow$} & \rotatebox[origin=c]{90}{$\subseteq$}_{f\;} \\
    X\ \ \ & \geetinf & \ Y.\ \ \ 
  \end{alignat*}
  \end{itemize}
\end{lemma}
\proof \quad
  \begin{itemize}
  \item Distribution of $\leetinf$ over $\subseteq_{qf}$:\\
    We create $Y\in\BT$ such that $X\leetinf Y\subseteq V$ by co-induction (remark that, by Lemma~\ref{lm:invarianceEtinfRB}, we obtain $V\in\BTqf$):
    \begin{itemize}
    \item $X=\Omega$: put $Y=\Omega$.
    \item Otherwise: we have
      \begin{align*}
        \quad
        X &= \lambda x_1...x_{n}.y\ X_1\cdots X_{m}, &
        U &= \lambda x_1...x_{n}.y\ U_1\cdots U_{m}, &
        V &= \lambda x_1...x_{n+k}.y\ V_1\cdots V_{m+k}, 
      \end{align*}
      such that $X_i\subseteq_{qf}U_i\leetinf V_i$ for $i\le m$ and $x_{n+i}\leetinf V_{m+i}$ (thus $V_{m+i}\in\BTqf$) for $i\le k$. By co-induction hypothesis we have $(Y_i)_{i\le m}$ such that $X_i\leetinf Y_i\subseteq V_i$ for $i\le m$, we thus set $$Y=\lambda x_1...x_{n+k}.y\ Y_1\cdots Y_{m}\ V_{m+1}\cdots V_{m+k}.$$
    \end{itemize}
    % Remark that by $\Omega$-finiteness of $X$, this process is finite. % and since we introduce a finite number of $\Omega$'s in its construction at each steps, $V$ is $\Omega$-finite.
  \item Distribution of $\geetinf$ over $\subseteq_{qf}$:\\
    We create $Y\in\BT$ such that $X\geetinf Y\subseteq V$ by co-induction, then, by Lemma~\ref{lm:invarianceEtinfRB}, we obtain that~$V\in\BTqf$:
    \begin{itemize}
    \item $X=\Omega$: put $Y=\Omega$.
    \item Otherwise: we have
      \begin{align*}
        \quad
        X &= \lambda x_1...x_{n+k}.y\ X_1\cdots X_{m+k},  &
        U &= \lambda x_1...x_{n+k}.y\ U_1\cdots U_{m+k},  &
        V &= \lambda x_1...x_{n}.y\ V_1\cdots V_{m},
      \end{align*}
      such that $X_i\subseteq_{qf}U_i\geetinf V_i$ for $i\le m$ and $X_{m+i}\subseteq_{qf} U_{m+i}\geetinf x_{n+i}$ for $i\le k$. By co-induction hypothesis we have $(Y_i)_{i\le m+k}$ such that $X_i\geetinf Y_i\subseteq V_i$ for $i\le m$, and~$X_{m+i}\geetinf Y_{m+i}\subseteq x_{n+i}$ for $i\le k$;  we thus set 
      $$Y=\lambda x_1...x_{n+k}.y\ Y_1\cdots Y_{m}.$$
    \end{itemize}
    % Then $V$ is $\Omega$-finite since this process conserves the number of $\Omega$'s.
  \item Distribution of $\geetinf$ over $\subseteq_{f}$:\\
    We create $Y\in\BTf$ similarly to the previous case except that we proceed by induction on $X$:
    \begin{itemize}
    \item $X=\Omega$: put $Y=\Omega$.
    \item Otherwise: we have
      \begin{align*}
        \quad
        X &= \lambda x_1...x_{n+k}.y\ X_1\cdots X_{m+k},  &
        U &= \lambda x_1...x_{n+k}.y\ U_1\cdots U_{m+k},  &
        V &= \lambda x_1...x_{n}.y\ V_1\cdots V_{m},
      \end{align*}
      such that $X_i\subseteq_{f}U_i\geetinf V_i$ for $i\le m$ and $X_{m+i}\subseteq_f U_{m+i}\geetinf x_{n+i}$ for $i\le k$. By co-induction hypothesis we have $(Y_i)_{i\le m+k}$ such that $X_i\geetinf Y_i\subseteq_{f}V_i$ for $i\le m$, and~$X_{m+i}\geetinf Y_{m+i}\subseteq_f x_{n+i}$ for $i\le k$;  we thus set 
      $$Y=\lambda x_1...x_{n+k}.y\ Y_1\cdots Y_{m}.\eqno{\qEd}$$
    \end{itemize}
  \end{itemize}
\todo{counter ex. for the last case?}

\subsubsection{Interpretations of B\"ohm trees}\label{sssec:BT2int} \quad \newline%
B\"ohm trees can be seen as normal forms of infinite depth. As such, one can define an interpretation of B\"ohm trees in a model via fixponts. However, there is no {\em a priori} reason to choose one specific fixpoint. We will formalize the notion of interpretation of B\"ohm trees in Definition~\ref{def:BTCohInterpret}. Then, using the description of such fixpoints, we will see in Propsition~\ref{prop:compLattOfCofInt} that the set of interpretations forms a complete lattice.

The minimal interpretation, called the inductive interpretation (Def.~\ref{def:(co)indInt}), is the canonical choice and has been used often in the literature to describe the approximation property (Def.~\ref{def:approximationProp}). Roughly speaking, the approximation property states the coherence of the interpretation of terms and the inductive interpretation of B\"ohm trees.

The complete lattice of interpretations is richer than the sole inductive interpretation. Another canonical interpretation is the maximal one, called co-inductive interpretation (Def.~\ref{def:(co)indInt}). Unfortunately, no equivalent version of approximation property can be given for the co-inductive interpretation (more exactly, no \Kweb can satisfy it).

However, we can look for an interpretation that is both, as large as possible and with a useful notion of coherence with the \Lcalcul. We found the quasi-finite interpretation (Def.~\ref{def:BTinterpretQF}) that is basically the minimal interpretation whose restriction to quasi-finite B\"ohm trees corresponds to the co-inductive interpretation. The property stating the coherence of interpretations is the quasi-approximation property (Def.~\ref{def:quasi-approxTh}). We will see later on that, in the presence of the approximation property and extensionality, the quasi-approximation property is equivalent to hyperimmunity and to full abstraction for \Hst.

\begin{definition}\label{def:BTCohInterpret}
  Let $D$ be a \Kweb. We call \newdefpremsec{proto-interpretation}{of B\"ohm trees} any total function~$\llb - \rrb_*$ that maps elements $U\in\BT$ to initial segments of $D^{\FV(U)}\Rta D$ (where $\FV(U)$ denotes the free variables of $U$).

  An \newdefpremsec{interpretation}{of B\"ohm trees} is a proto-interpretation \newsymprem{\protect\llb.\protect\rrb_*}{for B\"ohm trees} respecting the following:
  \begin{itemize}
  \item The interpretation of $\Omega$ is always empty:% (this corresponds to the sensibility):
    $$ \llb\Omega\rrb_*^{\vec x}=\emptyset. $$
  \item The interpretation of an abstraction $\lambda y. U$ satisfies:
    $$ \llb\lambda y.U\rrb_*^{\vec x}=\{(\vec a,b\cons\alpha) \mid (\vec a b,\alpha)\in\llb U\rrb_*^{\vec x y}\}. $$
  \item The interpretation of a list of applications $x_i\ U_1\cdots U_n$ (for $n\ge 0$), satisfies:
    $$ \llb x_i\ U_1\cdots U_n\rrb_*^{\vec x} = \{(\vec a,\alpha) \mid \exists b_1\cons\cdots\cons b_n\cons\alpha\le \alpha'\in a_i, \forall j\le n,\forall \beta\in b_j,(\vec a,\beta)\in \llb U_j\rrb_*^{\vec x}\}$$
  \end{itemize}
\end{definition}

\todo{example?}

\begin{remark}
  The different interpretations coincide on finite B\"ohm trees, thus we can write $\llb X\rrb^{\bar x}$ \newsyminvinv{\protect\llb.\protect\rrb^{\protect\vec x}}{for finite B\"ohm trees} for any $X\in\BTf$ without ambiguity, independently of the interpretation. Moreover, if the model is sensible, $\llb X \rrb^{\vec x}$ is the same as the interpretation of $X$ considered as a $\lambda$-term (by replacing occurrences of $\Omega$ by the diverging term \Om).%  and if $M$ is the $\lambda$-term corresponding to $X$ (with a diverging term $\Om$ in place of $\Omega$'s), then $\llb X \rrb^{\vec x} = \llb M\rrb^{\vec x}$.
\end{remark}

The interpretations differ on the infinite B\"ohm trees. Fortunately, the set of interpretations forms a complete lattice.

\begin{prop} \label{prop:compLattOfCofInt}
  The poset of interpretations (with pointwise inclusion) is a complete lattice.% whose least element is the \newdeftheprem{interpretation}{of B\"ohm trees}{inductive} $\llb U\rrb^{\vec x}_{ind} =\bigcup_{V\subseteq U \atop V\text{ finite}}\llb V\rrb^{\vec x}$.\newsyminv{\protect\llb .\protect\rrb^{\protect\vec x}_{ind}}
\end{prop}
\begin{proof}
  We show that the set of the interpretation is the set of the fixpoints of a Scott-continuous function $\zeta$ on the complete lattice of proto-interpretations (with pointwise order).

  The function \newsym{\protect\zeta} maps a proto-interpretation $\llb .\rrb_*$ to the proto-interpretation $\llb .\rrb_{\zeta(*)}$ defined as follows:
  \begin{itemize}
  \item The interpretation of $\Omega$ is always empty:
    $$ \llb\Omega\rrb_{\zeta(*)}=\emptyset. $$
  \item The interpretation of $\lambda y. U$ is the same as for $\lambda$-terms:
    $$ \llb\lambda y.U\rrb_{\zeta(*)}^{\vec x}=\{(\vec a,b\cons\alpha) \mid (\vec a b,\alpha)\in\llb U\rrb_{*}^{\vec x y}\}. $$
  \item The interpretation of $x_i\ U_1\cdots U_n$ satisfies:
    $$ \llb x_i\ U_1\cdots U_n\rrb_{\zeta(*)}^{\vec x} = \{(\vec a,\alpha) \mid \exists b_1\cons\cdots\cons b_n\cons\alpha\le \alpha'\in a_i, \forall j\le n,\forall \beta\in b_j,(\vec a,\beta)\in \llb U_j\rrb_*^{\vec x}\},$$
  \end{itemize}
  The two first equations trivialy preserve any sup. And the third equation preserves the directed sup since all $b_j$ are finite. These three equations preserve the directed sups, so that $\zeta$ is continuous. It is folklore that the set of fixpoints of a Scott-continuous function form a complete lattice. 
\end{proof}

\begin{definition} \label{def:(co)indInt}
  The minimal interpretation is the \newdeftheprem{interpretation}{of B\"ohm trees}{inductive} \newsyminv{\protect\llb .\protect\rrb^{\protect\vec x}_{ind}}
  $$\llb U\rrb^{\vec x}_{ind} =\bigcup_{X\subseteq U \atop X\in\BTf}\llb X\rrb^{\vec x}.$$

\noindent The maximal interpretation is called the \newdeftheprem{interpretation}{of B\"ohm trees}{co-inductive} and denoted \newsym{\protect\llb .\protect\rrb^{\protect\vec x}_{coind}}.
\end{definition}
 
\begin{figure*}
  \caption{Intersection type system for B\"ohm trees. Notice that the intersection is hiddent in the membership condition in the first premise of $(BT\protect\dash \protect\at)$.}\label{fig:IntTyBT}
  \begin{center}
    \AxiomC{$\Gamma,x:a\vdash U:\alpha$}
    \RightLabel{{\scriptsize$(BT\protect\dash \lambda)$}}
    \UnaryInfC{$\Gamma\vdash \lambda x.U : a\cons\alpha$}
    \DisplayProof\vspace{0.5em}

    \AxiomC{$b_1\cons\cdots\cons b_n\cons\beta\in a$}
    \AxiomC{$\alpha\le\beta$}
    \AxiomC{$\forall i\le n,\forall\gamma\in b_i,\ \Gamma,x:a\vdash U_i:\gamma$}
    \RightLabel{{\scriptsize$(BT\protect\dash \protect\at)$}}
    \TrinaryInfC{$\Gamma, x:a\vdash x\ U_1\cdots U_n:\alpha$}
    \DisplayProof
  \end{center}
\end{figure*}

The idea of intersection types can be generalized to B\"ohm trees. We introduce in Figure~\ref{fig:IntTyBT} the corresponding intersection type system. There is no rule for $\Omega$ since it has an empty interpretation. Remark, moreover, that the rule $(BT\protect\dash \at)$ seems complicated, but is just the aggregation of rules~$(I\protect\dash id)$,~$(I\protect\dash weak)$,~$(I\protect\dash \le)$ and~$(I\protect\dash \at)$ of Figure~\ref{fig:tyLam}. The difference between the inductive and the co-inductive interpretations lies on the finiteness of the allowed derivations in this system.

\begin{prop}\label{prop:IntTyBT}
  Let $U$ be a B\"ohm tree, then:
  \begin{itemize}
  \item $(\vec a,\alpha)\in\llb U\rrb_{ind}^{\vec x}$ iff the type judgment $\vec x:\vec a\vdash U:\alpha$  has a finite derivation using the rules of Figure~\ref{fig:IntTyBT}.
  \item $(\vec a,\alpha)\in\llb U\rrb_{coind}^{\vec x}$ iff the type judgment $\vec x:\vec a\vdash U:\alpha$  has a possibly infinite derivation using the rules of Figure~\ref{fig:IntTyBT}.
  \end{itemize}
\end{prop}
\todo{proof!}
 
%The B\"ohm trees being universally accepted as sharing deep relation with the \Lcalcul and its head-reduction, it is normal that we look for conditions relating our semantics to them. In this regards, the approximation property is widely acknowledged as a natural property. 

\begin{definition}\label{def:approximationProp}
  We say that $D$ respects the \newdef{approximation property}, or that $D$ is \newdef{approximable}, if the interpretation of any term corresponds to the inductive interpretation of its B\"ohm tree, {\em i.e.} if the following diagram commutes:
%  \begin{multicols}{2}
%    $$ \forall M\in\Lambda, \llb M\rrb^{\vec x} = \llb \BT(M)\rrb^{\vec x}_{ind}. $$
%    \columnbreak
    \begin{center}
      \begin{tikzpicture}[description/.style={fill=white,inner sep=2pt},ampersand replacement=\&]
        \matrix (m) [matrix of math nodes, row sep=1em, column sep=5em, text height=1.5ex, text depth=0.25ex]
        { {\Lamb} \& \& D \\
          \& {\BT} \\};
        \path[->] (m-1-1) edge[] node[description] {$\llb.\rrb$} (m-1-3);
        \path[->] (m-1-1) edge[] node[description] {$\BT(.)$} (m-2-2);
        \path[->] (m-2-2) edge[] node[description] {$\llb.\rrb_{ind}$} (m-1-3);
      \end{tikzpicture}
    \end{center}
%  \end{multicols}
\end{definition}
 
\begin{lemma}\label{lemma:incOfInterIfGeet}
  If $D$ is extensional and approximable, and if $M$ and $N$ are two terms such \linebreak that $M\geetinf N$ (def.~\ref{def:etainf}), then $\llb M \rrb^{\vec x}\subseteq \llb N\rrb^{\vec x}$.
\end{lemma}
\begin{proof}
  Let $(\vec a,\alpha)\in\llb M\rrb^{\vec x}$, by the approximation property there is a finite~$U\subseteq_f \BT(M)$ such that~$(\vec a,\alpha)\!\in\!\llb U\rrb^{\vec x}$. Since $U\!\subseteq_f\! \BT(M)\!\geetinf\! \BT(N)$, we can apply Lemma~\ref{lemma:distrEtinfSubset} to find $V\!\in\!\BTf$~such that $U\!\geetinf\! V\!\subseteq_f\!\BT(N)$. However, between finite B\"ohm trees, an $\infty\eta$-expansion is a usual~$\eta$-expan-sion, so that $U\!\succeq_\eta\! V\!\subseteq_f\!\BT(N)$. We thus have (using extensionality), $(\vec a,\alpha)\!\in\!\llb U\rrb^{\vec x}\!=\!\llb V\rrb^{\vec x}\!\subseteq\!\llb M\rrb^{\vec x}$ because the model is extensional.
\end{proof}

  The approximation property is a common condition enjoyed by all known K-models.\footnote{Provided that  they equalize terms with the same B\"ohm trees (which is a necessary condition for full abstraction).} %it is not a sufficient one to say anything regarding the full abstraction property.

\begin{example}\label{ex:ApprThe} 
  All the K-models of Example~\ref{example:1} except $P_\infty$ (that is not even sensible) are approximable, regardless of them being fully abstract or not. \todo{ref to the article one writen}
\end{example}
%\begin{proof}
%  By Example~\ref{ex:WP} and Corollary~\ref{cor:WP->ApTh}.
%\end{proof}

Our goal is to modify our set of approximants so that we could characterize the full abstraction.

\begin{remark} \label{rm:coindInt} 
  %One could use the interpretation given by the greatest fixpoint of  interpretations, called the \newdeftheprem{interpretation}{of B\"ohm trees}{co-inductive} and denoted $\llb .\rrb^{\vec x}_{coind}$.\newsyminv{\protect\llb .\protect\rrb^{\vec x}_{coind}} 
  A vain attempt would consist on replacing the inductive interpretation (in the definition of the approximation property) by the co-inductive one. The diagram of Definition~\ref{def:approximationProp} would never commute:

  For any sensible K-model and any $\alpha\!\in\! D$, if $M=\FixP\ (\lambda u.z\ u)$, then 
     \begin{align*}
      (\{\{\alpha\}\cons\alpha\},\alpha)&\in\llb \BT(M)\rrb^z_{coind\!\!\!}
      &(\{\{\alpha\}\cons\alpha\},\alpha)&\not\in\llb M\rrb^z.
    \end{align*}
  Indeed, if $(\{\{\alpha\}\cons\alpha\},\alpha)\in\llb M\rrb^z$ it would give $\alpha\in \llb M[\I/z]\rrb=\llb\FixP\ \I\rrb=\emptyset$. Moreover, \linebreak since~$\BT(M)= z\ \BT(M)$, we co-inductively get that $(\{\{\alpha\}\cons\alpha\},\alpha)\in\llb \BT(M)\rrb^z_{coind}$.

  In this example, the co-inductive interpretation of $\BT(\FixP\ (\lambda ux.z\ u))$ is incoherent with the term interpretation because it uses the $z$ infinitely often.\footnote{Notice that in a relational model \cite{EhRe04} this issue would not hold (even if other problems would come later) since in any elements of the interpretation $(a,\alpha)\in\llb\lambda x.M\rrb$ the $a$ is a finite multiset which can only ``see'' a finite number occurences of $z$.} In order to get rid of this incoherence we can use a guarded fixpoint.
\end{remark} \todo{deplacer une partie des exemples sur les types ac intersection}

In order to recover a meaningful property, we will use the {\em quasi-finite interpretation}. This is the least interpretation whose restriction to quasi-finite B\"ohm trees is the co-inductive interpretation.

\begin{definition}\label{def:BTinterpretQF}
  The \newdefthepremsec{interpretation}{of B\"ohm trees}{quasi-finite} is defined by \newsyminv{\protect\llb .\protect\rrb_{qf}} 
  $$\llb U\rrb_{qf}^{\vec x}=\bigcup_{X\subseteq U \atop X\in\BTqf}\llb X\rrb_{coind}^{\vec x}.$$
\end{definition}

\begin{definition}\label{def:quasi-approxTh}
  We say that $D$ respects the \newdef{quasi-approximation property}\newdefinvinv{approximation property}{quasi-}, or is \newdef{quasi-approximable}\newdefinvinv{approximable}{quasi-}, if the interpretation of any term corresponds to the quasi-finite interpretation of its B\"ohm tree, {\em i.e.} if the following diagram commutes:
  %\begin{multicols}{2}
  %  $$ \forall M\in\Lambda, \llb M\rrb^{\vec x} = \llb \BT(M)\rrb^{\vec x}_{qf}. $$
  %  \columnbreak
    \begin{center}
      \begin{tikzpicture}[description/.style={fill=white,inner sep=2pt},ampersand replacement=\&]
        \matrix (m) [matrix of math nodes, row sep=1em, column sep=5em, text height=1.5ex, text depth=0.25ex]
        { {\Lamb} \& \& D \\
          \& {\BT} \\};
        \path[->] (m-1-1) edge[] node[description] {$\llb.\rrb$} (m-1-3);
        \path[->] (m-1-1) edge[] node[description] {$\BT(.)$} (m-2-2);
        \path[->] (m-2-2) edge[] node[description] {$\llb.\rrb_{qf}$} (m-1-3);
      \end{tikzpicture}
    \end{center}
  %\end{multicols}
\end{definition}

\begin{example}
  We will prove that the quasi-approximation property is equivalent to hyperimmunity and full abstraction for \Hst (in presence of approximation property and extensionality). So models that are hyperimmune, like \Dinf, respect it and those that are not, like $\Dinf^*$, do not. In the case of~$\Dinf^*$, for example, the quasi-approximation property is refuted by \J, indeed $p\in \llb \BT(\J)\rrb_{qf}-\llb \J\rrb$.
\end{example}

\begin{remark}
  Notice that in general, approximability and quasi-approximability are independent (in the sense that none implies the other).
\end{remark}

\subsubsection{Technical lemma} \quad \newline%
This section shows that the relation %s $\subseteq$ and 
$\leetinf$ in \BT is pushed along the co-inductive interpretation into %the inclusion (Lemma.~\ref{lemma:subseteqBTtoInt}) and 
equality at the level of the model. %(Lemma.~\ref{lemma:leetinfBTtoInt})
This property will be useful as it generalizes easily to the quasi-finite interpretation.

\begin{lemma}\label{lemma:leetinfBTtoInt}
  Let $D$ be an extensional K-model and let $U,V$ be two B\"ohm trees such that $U\leetinf V$.\\
  Then $\llb U\rrb^{\vec x}_{coind}=\llb V\rrb^{\vec x}_{coind}$.
\end{lemma}
\begin{proof}
  We will prove separately the two inclusions.
  \begin{itemize}
  \item 
  We will show that the proto-interpretation $\llb V\rrb_*^{\vec x}=\bigcup_{U\geetinf V}\llb U\rrb_{coind}^{\vec x}$ over B\"ohm trees is an interpretation. This is sufficient since, $\llb\_\rrb_{coind}$ being the greatest interpretation, we will have 
  $$\llb V\rrb_{coind}\subseteq \bigcup_{U\geetinf V}\llb U\rrb_{coind}^{\vec x} =\llb V\rrb_* \subseteq\llb V\rrb_{coind}.$$
    \begin{itemize}
    \item Interpretation over $\Omega$:
      \begin{align*}
        \llb \Omega\rrb_*^{\vec x} & = \bigcup_{U\geetinf \Omega}\llb U\rrb_{coind}^{\vec x} = \llb \Omega \rrb_{coind}^{\vec x} = \emptyset.
      \end{align*}
    \item Otherwise:
      \begin{align*}
        \hspace{2em} 
&\llb \lambda x_{n+1}\dots x_{s}.x_j\ V_1 \cdots V_{k} \rrb_*^{\vec x} \\
        ={}&  \bigcup_{U\geetinf \lambda x_{n+1}...x_{s}.x_j\ V_1 \cdots V_{k}}\llb U \rrb_{coind}^{\vec x} \\
        ={}&  \bigcup_{m}\bigcup_{(U_i\geetinf V_i)_{i\le k}}\bigcup_{U_{k+i}\geetinf x_{s+i}}\llb \lambda x_{n+1}...x_{s+m}.x_j\ U_1 \cdots U_{k+m} \rrb_{coind}^{\vec x} \\
        ={}& \bigcup_{m}\bigcup_{U_i\geetinf V_i}\bigcup_{U_{k+i}\geetinf x_{s+i}} \{((a_i)_{i\le n},a_{n+1}\cons\cdots a_{s+m}\cons\alpha) \mid \exists c_1\cons\cdots\cons c_{k+m}\cons\alpha\le\alpha'\in a_j,\\
        & \hspace{21.5em} \forall t\le k\+m,\forall\beta\in c_t, \ (\vec a,\beta)\in \llb U_t\rrb_{coind}^{\vec x^{s+m}} \} \\
        ={}& \bigcup_{m}\{(a_i)_{i\le n},a_{n+1}\cons\cdots a_{s+m}\cons\alpha)\ \mid \exists c_1\cons\cdots\cons c_{k+m}\cons\alpha\le\alpha'\in a_j, \\
        & \hspace{14em} \forall t\le k,\forall\beta\in c_t,\ (\vec a,\beta)\in \bigcup_{U_t\geetinf V_t}\llb U_t\rrb_{coind}^{\vec x^{s+m}} \\
        & \hspace{14em} \forall t\le m,\forall\beta\in c_{k+t},\ (\vec a,\beta)\in \bigcup_{U_{k+t}\geetinf x_{s+t}}\llb U_t\rrb_{coind}^{\vec x^{s+m}} \} \\
        ={}& \bigcup_{m}\{(a_i)_{i\le n},a_{n+1}\cons\cdots a_{s+m}\cons\alpha)\ \mid \exists c_1\cons\cdots\cons c_{k+m}\cons\alpha\le\alpha'\in a_j, \\
        & \hspace{14em} \forall t\le k,\forall\beta\in c_t,\ (\vec a,\beta)\in \llb V_t\rrb_*^{\vec x^{s+m}} \\
        & \hspace{14em} \forall t\le m,\forall\beta\in c_{k+t},\ (\vec a,\beta)\in \llb x_{s+t}\rrb_*^{\vec x^{s+m}} \} \\
        ={}& \bigcup_{m}\{(a_i)_{i\le n},a_{n+1}\cons\cdots a_{s+m}\cons\alpha)\ \mid \exists c_1\cons\cdots\cons c_{k+m}\cons\alpha\le\alpha'\in a_j, \\
        & \hspace{14em} \forall t\le k,\forall\beta\in c_t,\ (\vec a,\beta)\in \llb V_t\rrb_*^{\vec x^{s+m}} \} 
      \end{align*}
    \end{itemize}
    This proves that if $U\geetinf V$, then $\llb U\rrb_{coind}^{\vec x}\subseteq \llb V\rrb_*^{\vec x}\subseteq \llb V\rrb_{coind}^{\vec x}$.
  \item To prove the converse, it is sufficient to show that the proto-interpretation $\llb V\rrb_*=\bigcup_{U\leetinf V}\llb V\rrb_{coind}$ is an interpretation:
    \begin{itemize}
    \item Interpretation over $\Omega$: 
      \begin{align*}
        \llb \Omega\rrb_*^{\vec x} ={}& \bigcup_{U\leetinf \Omega}\llb U\rrb_{coind}^{\vec x} = \llb\Omega\rrb_{coind}^{\vec x} = \emptyset.
      \end{align*}
    \item If $V_{s}\not\geetinf x_{k}$ and $V_{s+i}\geetinf x_{k+i}$ (for $1\le i\le m$) and $j\le k$: 
      \begin{align*}
        \hspace{2em}
&\llb\lambda x_{n+1}...x_{k+m}.x_j\ V_1 \cdots V_{s+m} \rrb_*^{\vec x^n} \\
        ={}&  \bigcup_{U\leetinf \lambda x_{n+1}...x_{k+m}.x_j\ V_1 \cdots V_{s+m}}\llb U\rrb_{coind}^{\vec x^n} \\
        ={}&  \bigcup_{m'\le m}\bigcup_{U_t\leetinf V_t}\llb\lambda x_{n+1}...x_{k+m'}.x_j\ U_1 \cdots U_{s+m'}\rrb_{coind}^{\vec x^n} \\
        ={}& \bigcup_{m'\le m}\bigcup_{U_t\leetinf V_t} \{(a_i)_{i\le n},a_{n+1}\cons\cdots a_{k+m'}\cons\alpha) \mid \exists c_1\cons\cdots\cons c_{s+m'}\cons\alpha\le\alpha'\in a_j,\\
        & \hspace{19em} \forall t\le s\+m',\forall\beta\in c_t,\ (\vec a,\beta)\in \llb U_t\rrb_{coind}^{\vec x^{k+m'}} \} \\
        ={}& \bigcup_{m'\le m}\{(a_i)_{i\le n},a_{n+1}\cons\cdots a_{k+m'}\cons\alpha) \mid \exists c_1\cons\cdots\cons c_{s+m'}\cons\alpha\le\alpha'\in a_j,\\
        & \hspace{16em} \forall t\le s\+m',\forall\beta\in c_t,\ (\vec a,\beta)\in \bigcup_{U_t\leetinf V_t}\llb U_t\rrb_{coind}^{\vec x^{k+m'}} \} \\
        ={}& \bigcup_{m'\le m}\{(a_i)_{i\le n},a_{n+1}\cons\cdots a_{k+m'}\cons\alpha)  \mid \exists c_1\cons\cdots\cons c_{s+m'}\cons\alpha\le\alpha'\in a_j, \\
        & \hspace{16em} \forall t\le s\+m',\forall\beta\in c_t,\ (\vec a,\beta)\in \llb V_t\rrb_*^{\vec x^{k+m'}} \} \\
        ={}& \bigcup_{m'\le m}\{(a_i)_{i\le n},a_{n+1}\cons\cdots a_{k+m}\cons\alpha) \mid \exists c_1\cons\cdots\cons c_{s+m'}\cons a_{k+m'+1}\cons\cdots a_{n+m}\cons\alpha\le\alpha'\in a_j,\\
        & \hspace{16em} \forall t\le s\+m',\forall\beta\in c_t,\ (\vec a,\beta)\in \llb V_t \rrb_*^{\vec x^{k+m'}}\} \\
        ={}& \bigcup_{m'\le m}\{(a_i)_{i\le n},a_{n+1}\cons\cdots a_{k+m}\cons\alpha) \mid \exists c_1\cons\cdots\cons c_{s+m'}\cons a_{k+m'+1}\cons\cdots a_{k+m}\cons\alpha\le\alpha'\in a_j, \\
        & \hspace{16em} \forall t\le s\+m', \forall\beta\in c_t,\ (\vec a,\beta)\in \llb V_t \rrb_*^{\vec x^{k+m'}}\\
        & \hspace{16em} \forall m'\le t\le m, \forall \beta\in a_{k+t}, (\vec a,\beta)\in \llb x_{k+t}\rrb_*^{\vec x^{k+m'}}\} \\
        ={}& \bigcup_{m'\le m}\{(a_i)_{i\le n},a_{n+1}\cons\cdots a_{k+m}\cons\alpha) \mid \exists c_1\cons\cdots\cons c_{s+m'}\cons a_{k+m'+1}\cons\cdots a_{k+m}\cons\alpha\le\alpha'\in a_j, \\
        & \hspace{16em} \forall t\le s\+m', \forall\beta\in c_t,\ (\vec a,\beta)\in \llb V_t \rrb_*^{\vec x^{k+m'}}\\
        & \hspace{16em} \forall m'\le t\le m, \forall \beta\in a_{k+t}, (\vec a,\beta)\in \llb V_{s+t}\rrb_*^{\vec x^{k+m'}}\}
      \end{align*}
    \end{itemize}
    This proves that if $U\leetinf V$, then $\llb U\rrb_{coind}^{\vec x}\subseteq \llb V\rrb_*^{\vec x}\subseteq \llb V\rrb_{coind}^{\vec x}$.\qedhere
  \end{itemize}
\end{proof}

  \subsection{Hyperimmunity implies full abstraction} \quad \newline\label{ssec:HtoFA1}%
% !TEX root = main.tex
% 
%
In this section we will prove the step $\eqref{eq:n1}\Rta\eqref{eq:n2}$ of the main theorem (Th.~\ref{th:final}). This will be done using the quasi-approximation property to decompose the proof into two steps. Indeed, we will see that in the presence of the approximation property, hyperimmunity implies the quasi-approximation property that itself implies the full abstraction for \Hst. Those two implications will be proved separately in Theorems~\ref{th:Hy+Ap=QAp} and~\ref{th:QApp->FullCompl}.

\subsubsection{Hyperimmunity and approximation imply quasi-approximation} \quad \newline%
Firstly, we are introducing tree-hyperimmunity that is equivalent to hyperimmunity (Lemma~\ref{lemma:sHyp<=>Hyp}).

The reason to introduce this new formalism is quite simple. For the proof of Theorem~\ref{th:Hy+Ap=QAp}, we will have to contradict hyperimmunity starting from a term $M$ that contradicts quasi-approximability. 

Recall that refuting hyperimmunity amounts to exhibiting a non-hyperimmune function ({\em i.e.}, bounded by a recursive function $g$) and a sequence $(\alpha_i)_i\in D^\Nat$ with a non well founded chain bounded by $g$ (see Definition~\ref{def:hyperim}).

The refutation of quasi-approximability by $M$ gives a recursive procedure that bounds the non-hyperimmune function $g$. However, the procedure does generally not directly construct the values of this function, but also performs a lot of useless computation; this is due to the refuting term $M$ not being optimal. Thus, we will simply construct an infinite tree and use K\"onig lemma\footnote{K\"onig lemma states that any infinite tree that is finitely branching accepts an infinite branch/path.} to find an infinite branch that contradicts hyperimmunity.% The existence of this tree is basically the contradiction of the tree-hyperimmunity.

Generalizing hyperimmunity from sequences to trees allows us to apply a well-known theorem of recursion theory. This theorem states the equivalence between hyperimmune functions and infinite paths in recursive $\Nat$-labeled trees.\footnote{Trees with nodes labeled by natural numbers.} That is why we can generalise hyperimmune functions to infinite recursive $\Nat$-labeled trees. The sequence $(\alpha_i)_i\in D^\Nat$, similarly, becomes a partial (but infinite) labeling of the recursive tree. The sequence has to be partial in order to select a specific hyperimmune path.

\begin{definition}
   Let $D$ be a \Kweb.\\
   A \newdef{\Nat-labeled} tree $T$ is a finitely branching tree where nodes are labeled by \Nat, we denote by $T(\mu)$ the~\Nat-label of the node $\mu$ in $T$.\\
  A \newdef{$D$-decoration} of a \Nat-labeled $T$ is a partial function of infinite domain $\partial_D : T\rta D$ such that for every couple of nodes $\nu$ and $\mu$ that are father and son in $T$, if $\mu\in dom(\partial_D)$, then $\nu\in dom(\partial_D)$ and:  
  \begin{align*}
    \partial_D(\nu)={}& a_{1}\cons\cdots\cons a_{T(\mu)}\cons\alpha & \Rta & & \partial_D(\mu)&\in a_{T(\mu)}.
  \end{align*}
  A K-model $D$ is \newdef{tree-hyperimmune} if none of the $\Nat$-labeled and $D$-decorated tree is recursive.
%
%  An extensional \Kweb $D$ is said to be \newdefsecprem{hyperimmune}{strongly} if for every infinite, finitely branching, \Nat-labeled, and computable tree $T$, there is no partial infinite labeling $\partial_D$ from nodes of $T$ to $D$ (non necessary computable) such that for all node $(\mu,n)\in dom(\partial_D)$, its father $(\nu,n')$ is also in the domain of $\partial_D$ and:
%  \begin{align*}
%    \partial_D(\nu,n')={}& a_{1}\cons\cdots\cons a_{n}\cons\alpha & \Rta & & \partial_D(\mu,n)&\in a_{n}.
%  \end{align*}
\end{definition}

\begin{lemma}\label{lemma:sHyp<=>Hyp}
  A \Kweb $D$ is tree-hyperimmune iff it is hyperimmune.
\end{lemma}
\begin{proof}\quad
  \begin{itemize}
  \item We assume that there is a recursive $g$ and a sequence $(\alpha_n)_n$ refuting hyperimmunity. We define the tree $T$ given by the set of nodes $\{\omega\in \Nat^*\mid \forall n\le|\omega|,\omega_n\le g(n)\}$ of finite sequences bounded by $g$ and ordered by prefix; the \Nat-labeling is given by~$T(\epsilon)=0$ and~$T(\omega.n)=g(n)$. Then $T$ is recursive and we have $\partial_D$ partially defined by induction:
    \begin{itemize}
    \item $\partial_D(\epsilon)=\alpha_0$ is always defined,
    \item $\partial_D(\omega.n)=\alpha_{|\omega.n|}$ is defined if $\partial_D(\omega)=\alpha_{|\omega|}=a_1\cons\cdots\cons a_n\cons \alpha$ and $\alpha_{|\omega.n|}=\alpha_{|\omega|+1}\in a_n$.
    \end{itemize}
  The decoration is infinite since, for all depth $d$, $\alpha_{d+1}\in \bigcup_{n\le g(d)}a_n$ for $\alpha_d=a_1\cons\cdots a_{g(d)}\cons\alpha_d'$. This contradicts tree-hyperimmunity.\todo{dessin?}
  \item If $D$ is not tree-hyperimmune, then there is a finitely branching, \Nat-labeled, and recursive tree $T$ and an infinite decoration $\partial_D$.
  By K\"onig lemma, the sub-tree that constitutes the domain of $\partial_D$ (which is infinite and finitely branching) accepts an infinite branch $(\mu_n)_n$. We denote ${\alpha_n:=\partial_D(\mu_n)}$, so that $\alpha_{n+1}\in a_{T(\mu_{n+1})}$ for $\alpha_n=a_1\cons\cdots\cons a_{T(\mu_{n+1})}\cons\alpha'$. Since the sequence $(T(\mu_{n+1}))_n$ is majored by the maximal \Nat-label on depth $n\+1$ in $T$, that is recursive, we are contradicting hyperimmunity.
  \end{itemize}
\end{proof}

\begin{remark}
  In the following, internal nodes of a quasi-finite B\"ohm tree are denoted by $X,Y...$ as they are idzntified with the quasi-finite B\"ohm tree whose root is the node at issue.
\end{remark}
\todo{Change the notation of nodes}

We now introduce the notion of the play of a quasi-finite B\"ohm tree $X$. The play of $X$ can be seen as the game semantics' play over the infinite arena $*{=}*\cons*$ performed by the execution of~$X$. Formally, it is a (possibly infinite) tree which father-son relationship corresponds to justification pointers. Moreover, players and opponents are playing alternatively, so that nodes at even depth are player nodes and play over applications, and nodes at odd depth are opponent nodes and play over abstractions. We will see that plays over quasi-finite B\"ohm trees remains finitely branching and recursive trees. Later on, we will try to decorate those plays to contradict tree-hyperimmunity.

\begin{definition}
  Let $X$ be a closed\footnotemark and recursive quasi-finite B\"ohm tree.\\
  The \newdef{play of $X$} is the recursive and $\Nat$-labeled tree $T$ whose nodes are of two kinds: %denoted $P(Y)$ (or $O(Y)$) are uniquely determined by a letter $P$ or $O$ and a node $Y$ over $X$. 
  \begin{itemize}
  \item The nodes at even depth are called \newdef{player nodes}. They are denoted $P(Y)$ for some $Y$ over $X$.
  \item The  nodes at odd depth are called \newdef{opponent nodes}. They are denoted $O(Y)$ for some $Y$ over $X$.
  \end{itemize}
  The tree is given by:
  \begin{itemize}
  \item the root is $P(X)$,
  \item the opponent node $O(\lambda x_1...x_m.z\ Y_1\cdots Y_k)$ has $k$ sons which are the $P(Y_i)$ for $i\le k$,
  \item the player node $P(\lambda x_1...x_m.z\ Y_1\cdots Y_k)$ has for sons every $O(Z)$ for $Z$ a node over $Y_1$, ..., or~$Y_n$ whose head variable is one of the $x_1,\dots,x_m$.
  \end{itemize}
\end{definition}\footnotetext{Can be generalised to non-closed trees by considering plays to be forests of trees.}

\begin{example}
  The tree below is the play over $\lambda x.x\ (\lambda yz. x\ (y\ z)\ (z\ y))$
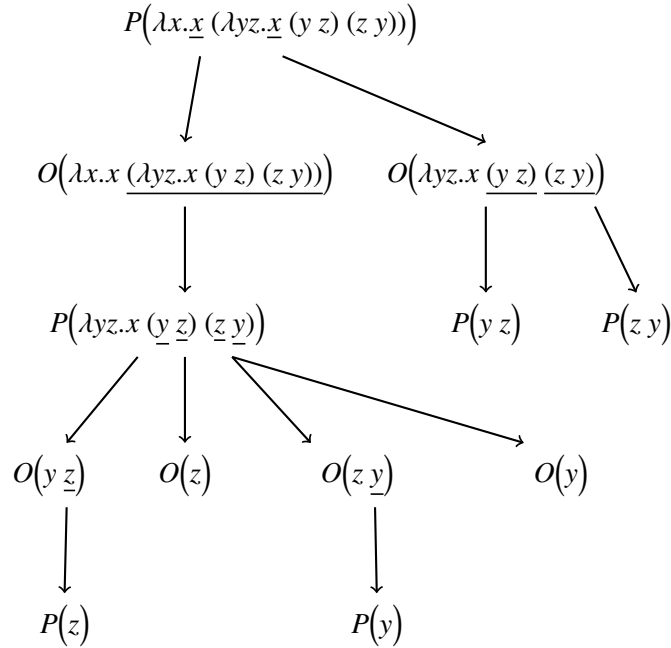
  \begin{center}
    \begin{tikzpicture}
      \node (P) at (1.75,8) {$\hspace{-3em}P\Bigl(\lambda x.\underline x\ (\lambda yz. \underline x\ (y\ z)\ (z\ y))\Bigr)\hspace{-5em}$};
      \node (O1) at (1,6) {$O\Bigl(\lambda x.x\ \underline{(\lambda yz. x\ (y\ z)\ (z\ y))}\Bigr)$};
      \node (O2) at (5,6) {$O\Bigl(\lambda yz. x\ \underline{(y\ z)}\ \underline{(z\ y)}\Bigr)\hspace{-0.7em}$};
      \node (P1) at (1,4) {$\hspace{-3.4em}P\Bigl(\lambda yz. x\ (\underline y\ \underline z)\ (\underline z\ \underline y)\Bigr)\hspace{-1.5em}$};
      \node (P2) at (5,4) {$P\Bigl(y\ z\Bigr)$}; 
      \node (P3) at (7,4) {$P\Bigl(z\ y\Bigr)$};
      \node (O11) at (-1,2) {$O\Bigl(y\ \underline z\Bigr)\hspace{-1em}$};
      \node (O12) at (1,2) {$O\Bigl(z\Bigr)$};
      \node (O13) at (3.1,2) {$O\Bigl(z\ \underline y\Bigr)\hspace{-1em}$};
      \node (O14) at (6,2) {$O\Bigl(y\Bigr)$};
      \node (P11) at (-0.6,0) {$P\Bigl(z\Bigr)$};
      \node (P12) at (3.55,0) {$P\Bigl(y\Bigr)$};
      \draw[->,thick] (P.south west) -- (O1.north);
      \draw[->,thick] (P.south east) -- (O2.north);
      \draw[->,thick] (O1.south) -- (P1.north);
      \draw[->,thick] (O2.south) -- (P2.north);
      \draw[->,thick] (O2.south east) --  (P3.north);
      \draw[->,thick] (P1.south west) --  (O11.north east);
      \draw[->,thick] (P1.south) --  (O12.north);
      \draw[->,thick] (P1.south east) -- (O13.north west);
      \draw[->,thick] (P1.south east) --  (O14.north west);
      \draw[->,thick] (O11.south east) --  (P11.north);
      \draw[->,thick] (O13.south east) -- (P12.north);
    \end{tikzpicture}
\end{center}
\end{example}

\begin{prop}
  Let $X$ be a closed and recursive quasi-finite B\"ohm tree and $T$ the play over $X$.
  For every node $Y$ of $X$, $P(Y)$ is a node of $T$. For every node $Y$ of $X$ that is not an $\Omega$, $O(Y)$ is a node of $T$.
\end{prop}\newpage
\begin{proof}
  By structural induction over the nodes $Y$ of $X$:
  \begin{itemize}
  \item If $Y$ is a node of $X$, then either $Y=X$ and $P(X)$ is the root of $T$, or $Y$ has a father $Y'$ in $X$. In the last case, $O(Y')$ is a node of $T$ by induction hypothesis and $P(Y)$ is a son of $O(Y')$.
  \item If $Y'=\lambda x_1...x_m.z\ Y_1\cdots Y_k$ is a node of $X$, then by closeness of $X$, there is an ancestor of $Y$ in~$X$ where $z$ is abstracted (potentially $Y=Y'$), {\em i.e}, $Y=\lambda y_1...y_{m'}.z'\ Y'_1\cdots Y'_k$ with $z=y_i$. By induction hypothesis, $P(Y)$ is a node of $T$ and $O(Y)$ is its son.\qedhere
  \end{itemize}
\end{proof}

\begin{definition}\label{def:labelEll}
  Let $X$ be a quasi-finite B\"ohm tree that is recursive and closed.\\
  The \newdef{labeled play over $X$} is the play over $X$ together with the~\Nat-labeling $\ell$ defined as follows:
  \begin{itemize}
  \item the labeling of the root is $\ell(P(X))=0$,
  \item any $Y$ at even depth, $P(Y)$, has for father $O(\lambda x_1...x_m.z\ Y_1\cdots Y_k)$ with $Y$ one of the $Y_i$, the~\Nat-label $\ell(P(Y))$ is the corresponding index of application $i$,
  \item any $Y=\lambda x_1...x_m.z\ Y_1\cdots Y_k$ at odd depth, $O(Y)$, has for father $P(Y')$ for $Y'$ that is the ancestor of $Y$ in $X$ where $z$ is abstracted (potentially $Y'=Y$), {\em i.e}, $Y'=\lambda y_1...y_{m'}.z'\ Y'_1\cdots Y'_k$ with $z=y_i$. The \Nat-label $\ell(O(Y))$ is the corresponding index of abstraction $i$.
  \end{itemize}
\end{definition}

\begin{example}
  The tree below is the labeled play over $X=\lambda x.x\ (\lambda yz. x\ (y\ z)\ (z\ y))$. For readability, the label is written in the parent-to-child arrow (we omit $\ell(X)=0$):
  \begin{center}
    \begin{tikzpicture}
      \node (P) at (1.75,8) {$\hspace{-3em}P\Bigl(\lambda x.\underline x\ (\lambda yz. \underline x\ (y\ z)\ (z\ y))\Bigr)\hspace{-5em}$};
      \node (O1) at (1,6) {$O\Bigl(\lambda x.x\ \underline{(\lambda yz. x\ (y\ z)\ (z\ y))}\Bigr)$};
      \node (O2) at (5,6) {$O\Bigl(\lambda yz. x\ \underline{(y\ z)}\ \underline{(z\ y)}\Bigr)\hspace{-0.7em}$};
      \node (P1) at (1,4) {$\hspace{-3.4em}P\Bigl(\lambda yz. x\ (\underline y\ \underline z)\ (\underline z\ \underline y)\Bigr)\hspace{-1.5em}$};
      \node (P2) at (5,4) {$P\Bigl(y\ z\Bigr)$}; 
      \node (P3) at (7,4) {$P\Bigl(z\ y\Bigr)$};
      \node (O11) at (-1,2) {$O\Bigl(y\ \underline z\Bigr)\hspace{-1em}$};
      \node (O12) at (1,2) {$O\Bigl(z\Bigr)$};
      \node (O13) at (3.1,2) {$O\Bigl(z\ \underline y\Bigr)\hspace{-1em}$};
      \node (O14) at (6,2) {$O\Bigl(y\Bigr)$};
      \node (P11) at (-0.6,0) {$P\Bigl(z\Bigr)$};
      \node (P12) at (3.55,0) {$P\Bigl(y\Bigr)$};
      \draw[->,thick] (P.south west) -- node [fill=white] {$1$} (O1.north);
      \draw[->,thick] (P.south east) -- node [fill=white] {$1$} (O2.north);
      \draw[->,thick] (O1.south) -- node [fill=white] {$1$} (P1.north);
      \draw[->,thick] (O2.south) -- node [fill=white] {$1$} (P2.north);
      \draw[->,thick] (O2.south east) -- node [fill=white] {$2$} (P3.north);
      \draw[->,thick] (P1.south west) -- node [fill=white] {$1$} (O11.north east);
      \draw[->,thick] (P1.south) -- node [fill=white] {$2$} (O12.north);
      \draw[->,thick] (P1.south east) -- node [fill=white] {$2$} (O13.north west);
      \draw[->,thick] (P1.south east) -- node [fill=white] {$1$} (O14.north west);
      \draw[->,thick] (O11.south east) -- node [fill=white] {$1$} (P11.north);
      \draw[->,thick] (O13.south east) -- node [fill=white] {$1$} (P12.north);
    \end{tikzpicture}
\end{center}
\end{example}

\begin{prop}
   For any quasi-finite $X\in\BTqf$, the labeled play $T$ over $X$ is recursive, finitely branching and \Nat-labeled.
\end{prop}
\begin{proof} 
  {\em The tree $T$ is finitely branching:} An opponent node $O(\lambda x_1\dots x_n.z\ Y_n\cdots Y_k)$ has exactly $k$ sons which are the $P(Y_i)$ for $i\le k$. A player node $P(\lambda x_1\dots x_n.z\ Y_n\cdots Y_k)$ has one son for each occurrence of its abstracted variables, which results in a finite number by quasi-finiteness of $X$.

  {\em The tree $T$ is recursive:} by recursivity and quasi-finiteness of $X$.
\end{proof}

\noindent Our objective is to $D$-decorate the labeled play of any quasi-finite B\"ohm tree $X$ such that $ \llb X \rrb_{coind} \neq \llb X\rrb_{ind} $. The $D$-decoration in question will follow a specific patern: we will furnish a path-D-decoration, which is a decoration of the nodes $\{P(Y_n),O(Y_n)\mid n\ge 0\}$ for $(Y_n)_{n\ge 0}$ a path in the B\"ohm tree of $X$.

\begin{definition}
  Let $D$ be a K-model and $X$ be a quasi-finite B\"ohm tree where all variables have been named differently.\\
  A path-$D$-decoration of the labeled play of $X$ is an infinite sequence $(Y_n)_{n\ge 0}$ of nodes of $X$ forming a path ({\em i.e.}, $Y_0=X$ and $Y_n$ father of $Y_{n+1}$) and three infinite sequences $(\alpha^n)_{n\ge 0},(\beta^n)_{n\ge 0}\in D^\Nat$ and~$(a_x)^{x\in\FV(Y_1,Y_2..)}$ such that for each $n$ (Where $\ell$ is the labeling of Definition~\ref{def:labelEll}):
  \begin{align*}
    \beta^n={}& b^n_{1}\cons\cdots\cons b^n_{\ell(P(Y_{n+1}))}\cons\beta' & \Rta & & \alpha^{n+1}&\in b^n_{\ell(P(Y_{n+1}))}.\\
    Y_n &= \lambda x_1,...,x_n. y\ X_1\cdots X_k &\Rta && \alpha^n&= a_{x_1}\cons\cdots a_{x_n}\cons\alpha'\\
    Y_n &= \lambda x_1,...,x_n. y\ X_1\cdots X_k &\Rta &&  \beta^n&\in a_{y}
  \end{align*}
\end{definition}

\begin{prop}\label{lm:path-decoration}
  Let $D$ be a K-model and $X$ be a quasi-finite B\"ohm tree.\\
  A path-$D$-decoration of the labeled play of $X$ induces a $D$-decoration of the labeled play of $X$.
\end{prop}
\proof
  Let $(Y_n)_{n\ge 0}$,  $(\alpha^n)_{n\ge 0}$ and $(\beta^n)_{n\ge 0}$ forming a path-$D$-decoration of the play of $X$. Then the partial function $\partial_D$ defined by $\partial_D(P(Y_n)):=\alpha^n$ and $\partial_D(O(Y_n)):=\beta^n$ for all $n$ is a $D$-decoration:
  \begin{itemize}
  \item the domain of $\partial_D$ is infinite since all $Y_n$ are different (they form a path),
  \item for any $n$, the father of $P(Y_{n+1})$ (decorated by $\alpha_n$) is $O(Y_n)$ which is decorated by $\beta_n$ and we have by hypothesis
    \begin{align*}
    \beta^n={}& b^n_{1}\cons\cdots\cons b^n_{\ell(P(Y_{n+1}))}\cons\beta' & \Rta & & \alpha^{n+1}&\in b^n_{\ell(P(Y_{n+1}))},
    \end{align*}
  \item for any $n$, the father of $O(Y_{n})$ is $P(Y_m)$ for some $m\le n$ such that the head variable $y$ of $Y_n$ is abstracted in the $\ell(O(Y_n))^{th}$ position in $Y_m$ and 
  \begin{align*}
    \alpha^m={}& a^m_{1}\cons\cdots\cons a^m_{\ell(O(Y_n))}\cons\alpha' 
    & \Rta && a^m_{\ell(O(Y_n))}&=a_y \\
     && \Rta & & \beta^n&\in a^m_{\ell(O(Y_n))}.
    \rlap{\hbox to 52 pt{\hfill\qEd}}
    \end{align*}
  \end{itemize}\medskip

\noindent What follows is a variant of K\"onig lemma where we are looking for an infinite path in $\BT(X)$ that we can decorate.

\begin{lemma}\label{lm:Hy+Ap=QAp}
  Let $D$ be a \Kweb and $X\in\BTqf$ be a quasi-finite B\"ohm tree. If 
   $$ \llb X \rrb_{coind} \neq \llb X\rrb_{ind}, $$
   then $D$ is not tree-hyperimmune.
\end{lemma}
\begin{proof}
  We can assume that $X$ is closed (otherwise we could have taken $\lambda x_1...x_m.X$)\\
  Let $\alpha\in \llb X \rrb_{coind} - \llb X\rrb_{ind}$.\\
  We define a path-$D$-decoration of the labeled play of $X$, breaking the conditions of tree-hyperimmunity by Lemma~\ref{lm:path-decoration}. For that we give, inductively, an infinite path $(Y_n)_n$ in $X$, and three infinite sequences $(\alpha^n)_{n\ge 0},(\beta^n)_{n\ge 0}\in D^\Nat$ and $(a_x)^{x\in\FV(Y_1,Y_2..)}$ forming the path-$D$-decoration. Moreover, those are defined such that for all $n$, $(\vec a,\alpha^{n})\in\llb Y_{n}\rrb_{coind}^{\vec x}-\llb Y_{n}\rrb_{ind}^{\vec x}$: 
  \begin{itemize}
  \item $Y_0=X$ and $\alpha^0=\alpha$.
  \item Assume that we got $Y_{n}$. By non emptiness of $\llb Y_n\rrb_{coind}^{\vec x}$, we have $Y_n=\lambda x_{1}...x_{m}.y\ X_1\cdots X_k$ with $x_1...x_m$ as free variables:\\
    If we unfold $a_{x_1}\cons\cdots a_{x_m}^{n}\cons \alpha' := \alpha^{n}$, then there exists $\beta^n = b^n_1\cons\cdots b^n_k\cons \alpha_0''\in a_y$ (with $\alpha''\ge\alpha'$) such that for all $j$ and all $\gamma\in b^n_j$, we have $(\vec a,\gamma)\in \llb X_j\rrb_{coind}^{\vec x}$.\\
    In particular there is $j\le k$ and  $\alpha^{n+1}\in b^n_j$ such that  $(\vec a,\alpha^{n+1})\in \llb X_j\rrb_{qf}^{\vec x}-\llb X_j\rrb_{ind}^{\vec x}$.\\
    We set $Y_{n+1} := X_j$ so that
    \begin{itemize}
    \item $\beta^n = b^n_1\cons\cdots b^n_k\cons \alpha_0''\in a_y$ and $\alpha^{n+1}=\gamma\in b^n_j=b^n_{\ell(P(X_j))}$,
    \item $Y_n=\lambda x_{1}...x_{m}.x_i\ X_1\cdots X_k$ and $\alpha^{n}= a_{x_1}\cons\cdots a_{x_m}^{n}\cons \alpha'$,
    \item $Y_n=\lambda x_{1}...x_{m}.x_i\ X_1\cdots X_k$ and $\beta^n\in a_{y}$.\qedhere
    \end{itemize}
  \end{itemize}
\end{proof}

\begin{theorem}\label{th:Hy+Ap=QAp}
  Any hyperimmune approximable \Kweb $D$ is also quasi-approximable.
\end{theorem}
\begin{proof}
  We will prove the contrapositive: We assume that $D$ is approximable but not quasi-approximable, then we show that $D$ is not hyperimmune.\\
  Since $D$ is not quasi-approximable, there is a $\lambda$-term $M\in\Lambda$ such that ${\llb M\rrb^{\vec x}\neq \llb \BT(M)\rrb_{qf}^{\vec x}}$.\\
  The approximation property gives that $\llb M\rrb^{\vec x}= \llb \BT(M)\rrb_{ind} \subset \llb \BT(M)\rrb_{qf}$. Thus there is a quasi finite $X\subseteq_{qf}\BT(M)$ such that $\llb X\rrb_{coind}\neq\llb X\rrb_{ind}$.\\
  By Lemma~\ref{lm:Hy+Ap=QAp}, the K-model $D$ is not tree-hyperimmune and thus not hyperimmune by Lemma~\ref{lemma:sHyp<=>Hyp}.
\end{proof}

\subsubsection{Quasi-approximation and extensionality imply full abstraction} %\quad \newline%
%
%\begin{lemma}[Sensibility]\label{lemma:sensQAp}
%  Let $D$ a K-model respecting the quasi-approximation property, then $D$ is sensible (diverging terms have empty interpretations).
%\end{lemma}
%\begin{proof}
%\end{proof}
\begin{theorem}\label{th:QApp->Adeq}
  Let $D$ be a K-model respecting the quasi-approximation property. Then it is inequationally adequate, {\em i.e.}, for all $M$ and $N$ such that $\llb M\rrb^{\vec x}\subseteq \llb N\rrb^{\vec x}$there is $M\leob N$.
\end{theorem}
\begin{proof}
  $D$ is sensible (diverging terms have empty interpretations). Indeed, for any head-diverging term $M$, $\BT(M)=\Omega$ and thus 
  $$\llb M\rrb^{\vec x}=\llb \BT(M)\rrb_{qf}^{\vec x}=\llb \Omega\rrb_{qf}^{\vec x}=\emptyset.$$
  We conclude since sensibility implies inequational adequacy.
%  Let $C\in\Lcont$ such that $C\llc M\rrc\Downarrow$. By Lemma~\ref{lemma:sensQAp}, $D$ is sensible and there is $(\vec a,\alpha)\in\llb C\llc M\rrc\rrb^{\vec x}$. By easy induction on the interpretation that is defined by context, and since $\llb M\rrb^{\vec x}\subseteq \llb N\rrb^{\vec x}$, we obtain that $(a,\alpha)\in \llb C\llc N\rrc\rrb^{\vec x}$. By applying the sensitivity of $D$ we have that $C\llc N\rrc$ head-converges.
\end{proof}

\begin{theorem}\label{th:QApp->FullCompl}
  Let $D$ be a quasi-approximable extensional K-model. $D$ is inequationally complete, {\em i.e.}, for all $M$ and $N$; $M\leob N$ implies $\llb M\rrb^{\vec x}\subseteq\llb N\rrb^{\vec x}$.
\end{theorem}
\proof
  Let $(\vec a,\alpha)\in \llb M\rrb^{\vec x}$.\\
  By the quasi-approximation property, there is $W\subseteq_{qf}\BT(M)$ such that $(\vec a,\alpha)\in \llb W\rrb_{coind}^{\vec x}$.\\
  By Proposition~\ref{prop:H*le}, there are $U$ and $V$ such that $\BT(M)\leetinf U \subseteq V \geetinf \BT(N)$. By applying Lemma~\ref{lemma:distrEtinfSubset} on $W\subseteq_{qf}\BT(M)\leetinf U$, we get $X\in\BTqf$ and by applying it a second time \linebreak on~$X\subseteq_{qf}V\geetinf \BT(N)$ we get $Y$ such that:
  \begin{alignat*}6
    \BT(M)\ & \leetinf & U\ \ \ & \subseteq & V\ \ \  & \geetinf &\ \BT(N) \\
    \rotatebox[origin=c]{90}{$\subseteq$}_{qf}\ \ \ & & \rotatebox[origin=c]{90}{$\subseteq$}_{qf} & & \quad \rotatebox[origin=c]{90}{$\subseteq$}_{qf} & & \rotatebox[origin=c]{90}{$\subseteq$}_{qf}\quad  \\
    W\quad\ & \leetinf & X\ \ \  ={}& & X\ \ \ & \geetinf & Y \quad\ \ 
  \end{alignat*}
  Thus:
  \begin{align*}
    (\vec a,\alpha)\ \in\ \llb W\rrb_{coind}^{\vec x} ={}& \llb X \rrb^{\vec x}_{coind}  && \text{Lemma~\ref{lemma:leetinfBTtoInt}}\\
    %& \subseteq \llb Y \rrb^{\vec x}_{coind}  && \text{by Lemma~\ref{lemma:subseteqBTtoInt}}\\
    ={}& \llb Y \rrb^{\vec x}_{coind} && \text{by Lemma~\ref{lemma:leetinfBTtoInt}}\\
    & \subseteq \llb N\rrb^{\vec x} && \text{by
                                       quasi-approximation}\rlap{\hbox
                                       to 66 pt{\hfill\qEd}}
  \end{align*}

\newpage
  \subsection{Full abstraction implies hyperimmunity}
    \label{ssec:FAtoH1}
% !TEX root = main.tex
% 

\subsubsection{The counterexample} \quad \newline%
Suppose that $D$ is approximable but is not hyperimmune. By Definition~\ref{def:hyperim} of hyperimmunity, there exists a recursive $g:\Nat\rta\Nat$ and a sequence {$(\alpha_n)_{n\ge 0}\in D^\Nat$} such that 
\begin{align*}
  \alpha_n&=a_{n,1}\cons\cdots\cons a_{n,g(n)}\cons \alpha_n' 
  & &\text{with} &
  \alpha_{n+1}&\in \bigcup_{k\le g(n)}a_{n,k}.
\end{align*}
We will use the function $g$ to define a term $\J_g$ (Eq.~\ref{eq:defA}) such that $(\J_g\ \underline{0})$ is observationally equal to the identity in $\Lamb$ (Lemma~\ref{lemma:ineq}) but can be denotationally distinguished in $D$ (Lemma~\ref{lemma:DenotationalSeparation}). This allows to refute full abstraction:

\begin{theorem}\label{th:countex}
  If $D$ is approximable but not hyperimmune, then it is not fully abstract for the~$\lambda$-calculus.
\end{theorem}

Basically, $(\J_g\ \underline{0})$ is a generalization of the term $\boldsymbol{J}$ used in \cite{CDZ87} to prove that the model $\Dinf^*$ (Ex.~\ref{example:1}) is not fully abstract. The idea is that $\boldsymbol{J}$ is the infinite $\eta$-expansion of the identity $\I$ where each level of the B\"ohm tree is $\eta$-expanded by one variable. Our term $(\J_g\ \underline{0})$ is also an infinite $\eta$-expansion of $\I$, but now, each level of the B\"ohm tree is $\eta$-expanded by $g(n)$ variables.\footnote{In the article \cite{bre13} of the same author, the reader may also find another counterexample based on the same kind of intuitions.}

\medskip

\noindent
Let $(\G_n)_{n\in \Nat}$ be the sequence of closed $\lambda$-terms defined by: 
\begin{equation}
 \newsym{\protect\G_n}:= \lambda ue x_1...x_{g(n)}.e\ (u\ x_1)\ \cdots\ (u\ x_{g(n)}) \label{eq:defG}
\end{equation}

\noindent The recursivity of $g$ implies the recursivity of the sequence $\G_n$. Thus, we can use Proposition~\ref{prop:autorecursivity}: there exists a $\lambda$-term \newsym{\protect\G} such that:
\begin{equation}
 \G\ \underline{n}\rta^* \G_n. \label{eq:redG}
\end{equation}

\noindent Recall that $\Succ$ denotes the Church successor function and $\FixP$ the Turing fixpoint combinator.\\ We define:
\begin{equation}
 \newsym{\protect\J_g} :=\FixP\ (\lambda uv. \G\ v\ (u\ (\Succ\ v))). \label{eq:defA}
\end{equation}
Then:
\begin{equation}
 \J_g\ \underline{n} \rta^* \G_n\ (\J_g\ \underline{n\+1}), \label{eq:redA}
\end{equation}
and its B\"ohm tree can be sketched as \todo{redraw}
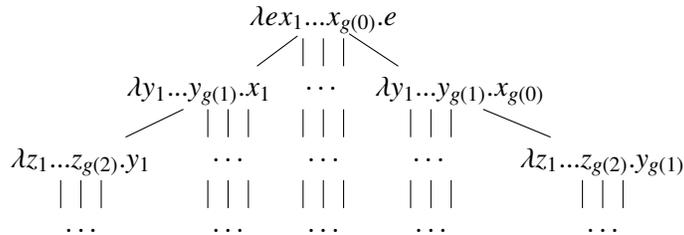
\begin{center}
    \begin{tikzpicture}[description/.style={fill=white,inner sep=2pt},ampersand replacement=\&]
      \matrix (m) [matrix of math nodes, row sep=1em, column sep=0.5em, text height=1.5ex, text depth=0.25ex]
      { \& \& \hspace{-2em}\lambda e x_1...x_{g(0)}. e\hspace{-2em}\\
        \& \hspace{-2em}\lambda y_1...y_{g(1)}. x_1  \& \cdots \& \lambda y_1...y_{g(1)}. x_{g(0)} \hspace{-2em}\\ 
        \lambda z_1...z_{g(2)}. y_1  \& \cdots \& \cdots \& \cdots \& \lambda z_1...z_{g(2)}. y_{g(1)}\\
        \cdots  \& \cdots \& \cdots \& \cdots \& \cdots\\ };
      \path[-] (m-1-3) edge[] node[auto] {} (m-2-2);
      \path[-] (m-1-3) edge[] node[auto] {} (m-2-4);
      \path[-] (m-2-2) edge[] node[auto] {} (m-3-1);
      \path[-] (m-2-4) edge[] node[auto] {} (m-3-5);

      \path[transform canvas={xshift = 0.7em}] (m-1-3) edge[] node[auto] {} (m-2-3);
      \path[-] (m-1-3) edge[] node[] {} (m-2-3);
      \path[transform canvas={xshift = -0.7em}] (m-1-3) edge[] node[auto] {} (m-2-3);
      \path[transform canvas={xshift = 0.7em}] (m-2-3) edge[] node[auto] {} (m-3-3);
      \path[-] (m-2-3) edge[] node[] {} (m-3-3);
      \path[transform canvas={xshift = -0.7em}] (m-2-3) edge[] node[auto] {} (m-3-3);
      \path[transform canvas={xshift = 0.7em}] (m-2-2) edge[] node[auto] {} (m-3-2);
      \path[-] (m-2-2) edge[] node[] {} (m-3-2);
      \path[transform canvas={xshift = -0.7em}] (m-2-2) edge[] node[auto] {} (m-3-2);
      \path[transform canvas={xshift = 0.7em}] (m-2-4) edge[] node[auto] {} (m-3-4);
      \path[-] (m-2-4) edge[] node[] {} (m-3-4);
      \path[transform canvas={xshift = -0.7em}] (m-2-4) edge[] node[auto] {} (m-3-4);
      \path[transform canvas={xshift = 0.7em}] (m-3-3) edge[] node[auto] {} (m-4-3);
      \path[-] (m-3-3) edge[] node[] {} (m-4-3);
      \path[transform canvas={xshift = -0.7em}] (m-3-3) edge[] node[auto] {} (m-4-3);
      \path[transform canvas={xshift = 0.7em}] (m-3-2) edge[] node[auto] {} (m-4-2);
      \path[-] (m-3-2) edge[] node[] {} (m-4-2);
      \path[transform canvas={xshift = -0.7em}] (m-3-2) edge[] node[auto] {} (m-4-2);
      \path[transform canvas={xshift = 0.7em}] (m-3-4) edge[] node[auto] {} (m-4-4);
      \path[-] (m-3-4) edge[] node[] {} (m-4-4);
      \path[transform canvas={xshift = -0.7em}] (m-3-4) edge[] node[auto] {} (m-4-4);
      \path[transform canvas={xshift = 0.7em}] (m-3-1) edge[] node[auto] {} (m-4-1);
      \path[-] (m-3-1) edge[] node[] {} (m-4-1);
      \path[transform canvas={xshift = -0.7em}] (m-3-1) edge[] node[auto] {} (m-4-1);
      \path[transform canvas={xshift = 0.7em}] (m-3-5) edge[] node[auto] {} (m-4-5);
      \path[-] (m-3-5) edge[] node[] {} (m-4-5);
      \path[transform canvas={xshift = -0.7em}] (m-3-5) edge[] node[auto] {} (m-4-5);
    \end{tikzpicture}
\end{center}

\noindent Lemma~\ref{lemma:ineq} below proves that $\J_g\ \underline{0}$ is operationally equivalent to the identity \I. In fact it is an infinite $\eta$-expansion of \I. But first, we need the following auxiliary lemma.

\begin{lemma}\label{lemma:addz}
  For any terms $M,N\mathrm{\in}\Lamb$ and any fresh $z$:
  $$ (M\ z\geetinf N\ z)\quad \Rta\quad (M\geetinf N).$$
\end{lemma}
\begin{proof}
  If $M$ diverges, then so does $(M\ z)$, thus~$(N\ z)\Ua$ and $N\Ua$, so that $\BT(M)=\BT(N)=\Omega$.\\
  Otherwise we have $M\rta_h^* \lambda x_1\dots x_n. y\ M_1\cdots M_k$:
  \begin{itemize}
  \item If $n=0$, then $\ M\ z\rta_h^*y\ M_1\cdots M_k\ z\ $ and $\ N\ z\rta_h^*y\ N_1\cdots N_k\ z\ $ with $M_i\geetinf N_i$, thus~$M\geetinf N$.
  \item Otherwise, $M\ z\rta_h^*\lambda x_2\dots x_n.y[z/x_1]\ M_1[z/x_1]\cdots M_k[z/x_1]$\\
   and $\ N\ z \rta_h^*N'\preceq_\eta\lambda x_2\dots x_n.y[z/x_1] \ N_1\cdots N_k\ $ with $\ M_i[z/x_1]\geetinf N_i\ $ for all $i$.
   Thus, since~$z$ is fresh, $N \rta_h^*\lambda x_1.N'[x_1/z]\preceq_\eta\lambda x_1\dots x_n.y \ N_1[x_1/z]\cdots N_k[x_1/z]$ %\linebreak
    and $M_i\geetinf N_i[x_1/z]$, so $M\geetinf N$.\qedhere
  \end{itemize}
\end{proof}

\begin{lemma}\label{lemma:ineq}
  We have $\J_g\ \underline{0}\equivob \Id$.
\end{lemma}
\begin{proof}
  We prove that $(\J_g\ \underline{n}\ z)\geetinf z$ (where $z$ is fresh) for every $n$, by co-induction and unfolding of $\BT(\J_g\ \underline{n}\ z)$: 
  \begin{align*}
    & \boldsymbol{BT}(\J_g\ \underline{n}\ z) \\ 
    ={}& \boldsymbol{BT}(\boldsymbol{G}_n\ (\J_g\ \underline{n\+1})\ z) && \!\!\text{by \eqref{eq:redA}}\\
    ={}& \lambda \vec x^{g(n)}\!\!.z\ \boldsymbol{BT}(\J_g\ \underline{n\+1}\ x_1) \cdots \boldsymbol{BT}(\J_g\ \underline{n\+1}\ x_{g(n)}) && \!\!\text{by \eqref{eq:defG}}\\
    & \geetinf \lambda \vec x^{g(n)}.z\  x_1\cdots x_{g(n)} && \hspace{-1.8em}\text{by co-Ind}\\
    & \succeq_\eta z
  \end{align*}
  By applying Lemma~\ref{lemma:addz}, we know that $(\J_g\ \underline{n})\geetinf\Id$ and by Corollary~\ref{cor:geetinfImpEq} that $\J_g\ \underline{0}\equivob \Id$.
\end{proof}

\subsubsection{Denotational separation} \quad \newline
In this section we show  that $\J_g\ \underline{0}$ and $\Id$ are denotationally separated (Lemma~\ref{lemma:DenotationalSeparation}), despite being operationally equivalent.

Let $\newsym{\protect\J_g^{n,k}(z)}\in\BTf(\J_g\ \underline n\ z)$ be the truncation of $\BT(\J_g\ \underline n\ z)$ at depth $k$ (in particular $\J_g^{n,0}=\Omega$). 

\begin{example}
  For example, $\J_g^{5,3}(z)$ is the B\"ohm tree:

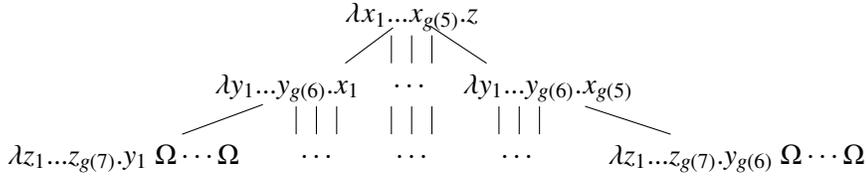
  \begin{tikzpicture}[description/.style={fill=white,inner sep=2pt},ampersand replacement=\&]
    \matrix (m) [matrix of math nodes, row sep=1em, column sep=0.5em, text height=1.5ex, text depth=0.25ex]
    { \& \& \hspace{-2em}\lambda x_1...x_{g(5)}. z\hspace{-2em}\\
      \& \hspace{-2em}\lambda y_1...y_{g(6)}. x_1  \& \cdots \& \lambda y_1...y_{g(6)}. x_{g(5)}\hspace{-2em}\\ 
      \lambda z_1...z_{g(7)}. y_1\ \Omega \cdots \Omega  \& \cdots \& \cdots \& \cdots \& \lambda z_1...z_{g(7)}. y_{g(6)}\ \Omega \cdots \Omega \\};
    \path[-] (m-1-3) edge[] node[auto] {} (m-2-2);
    \path[-] (m-1-3) edge[] node[auto] {} (m-2-4);
    \path[-] (m-2-2) edge[] node[auto] {} (m-3-1);
    \path[-] (m-2-4) edge[] node[auto] {} (m-3-5);

    \path[transform canvas={xshift = 0.7em}] (m-1-3) edge[] node[auto] {} (m-2-3);
    \path[-] (m-1-3) edge[] node[] {} (m-2-3);
    \path[transform canvas={xshift = -0.7em}] (m-1-3) edge[] node[auto] {} (m-2-3);
    \path[transform canvas={xshift = 0.7em}] (m-2-3) edge[] node[auto] {} (m-3-3);
    \path[-] (m-2-3) edge[] node[] {} (m-3-3);
    \path[transform canvas={xshift = -0.7em}] (m-2-3) edge[] node[auto] {} (m-3-3);
    \path[transform canvas={xshift = 0.7em}] (m-2-2) edge[] node[auto] {} (m-3-2);
    \path[-] (m-2-2) edge[] node[] {} (m-3-2);
    \path[transform canvas={xshift = -0.7em}] (m-2-2) edge[] node[auto] {} (m-3-2);
    \path[transform canvas={xshift = 0.7em}] (m-2-4) edge[] node[auto] {} (m-3-4);
    \path[-] (m-2-4) edge[] node[] {} (m-3-4);
    \path[transform canvas={xshift = -0.7em}] (m-2-4) edge[] node[auto] {} (m-3-4);
  \end{tikzpicture}
\end{example}

\noindent We recall that the sequence $(\alpha_n)_{n\ge 0}$, obtained from the refutation of the hyperimmunity, verifies~$\alpha_n=a_{n,1}\cons\cdots\cons a_{n,g(n)}\cons \alpha_n'$ with $\alpha_{n+1}\in \bigcup_{k\le g(n)}a_{n,k}$.

\begin{lemma}\label{lemma:separationOfAnk}
  For all $n$ and $k$, and for all $a\in\Achf{D}$ such that $\alpha_n\in a$, we have 
  $$(a,\alpha_n)\not\in\llb\J_g^{n,k}(z)\rrb_D^z.$$
\end{lemma}
\begin{proof}
  By induction on $k$:
  \begin{itemize}
  \item ($k\!=\!0$): since $\J_n^{n,0}=\Omega$ then by the approximation property we derive $(a,\alpha_n)\!\not\in\! \llb\J_g^{n,k}(z)\rrb_D^z\!=\!\emptyset$.
  \item ($k+1$):
    Remark that $\J_g^{n,k+1}(z)= \lambda x_1...x_{g(n)}.z\ (\J_g^{n+1,k}(x_1))\ \cdots\ (\J_g^{n+1,k}(x_{g(n)}))$ and that for all~$i$,~$x_i$ is the only free variable of $\J_g^{n+1,k}(x_i)$.\\
    We unfold $\alpha_n=a^n_1\cons\cdots\cons a^n_{g(n)}\cons\alpha_n'$.
    Then $(a,\alpha_n)$ belongs to $\llb\J_g^{n,k\+1}(z)\rrb_D^z$ iff there \linebreak is~$\beta=b_1\cons\cdots\cons b_{g(n)}\cons\alpha_n''\in a$ (with $\alpha''_n\ge\alpha'_n$) such that for all~$i\le g(n)$ and for  all~$\gamma\in b_i$, there is $(a_i^n,\gamma)\in \llb\J_g^{n+1,k}(x_i)\rrb_D^{x_i}$. The refutation has two cases:
    \begin{itemize}
    \item For $\beta=\alpha_n$: there is $i\le g(n)$ such that $\alpha_{n+1}\in b_i=a_i^n$, so that the induction hypothesis gives $(a_i^n,\alpha_{n+1})\not\in \llb\J_g^{n+1,k}(x_i)\rrb_D^{x_i}$.
    \item For $\beta\neq\alpha_n$, since $a$ is an anti-chain and $\alpha_n\in a$, $\beta\not\ge\alpha_n$. We have seen that $\alpha''_n\ge\alpha'_n$, thus, there is $i\le g(n)$ such that $b_i\not\le a^n_i$. In particular, there is $\gamma\in b_i$ such that $\gamma\not\le\delta$ for any $\delta\in a_i^n$, thus $(a_i^n,\gamma)\not\in\llb \I\ x_i\rrb_D^{x_i}$. Since $(\I\ x_i)\leetinf (\J_g\ \underline{n\+1}\ x_i)$, by applying Lemma~\ref{lemma:incOfInterIfGeet} we obtain $\llb \I\ x_i\rrb_D^{x_i} \supseteq \llb \J_g\ \underline{n\+1}\ x_i\rrb_D^{x_i}\supseteq \llb \J_{n+1}^{n+1,k}(x_i)\rrb_D^{x_i}$.\qedhere
    \end{itemize}
  \end{itemize}
\end{proof}

\begin{lemma}\label{lemma:DenotationalSeparation}
  The term $\J_g\ \underline{n}$ (for any $n$) and the identity are denotationally separated in $D$:
  \[ \llb \J_g\ \underline n\rrb_D \quad \neq \quad \llb \I \rrb_D\]
\end{lemma}
\begin{proof}
  Using the approximation property and extensionality, it is sufficient to prove that 
  $$\{\alpha_0\}\cons\alpha_0\not\in \bigcup_k\llb\lambda z.\J_g^{n,k}(z)\rrb_D=\bigcup_{U\in\BTf(\J_g\ \underline n)}\llb U\rrb_D=\llb \J_g\ \underline n\rrb_D,$$
  which can be obtained by the application of Lemma~\ref{lemma:separationOfAnk}.
\end{proof}

\noindent This concludes the proof of the main theorem (Theorem \ref{th:final}):

\smallskip

\begin{em}
  For any extensional approximable \Kweb $D$, the following are equivalent:
  \begin{enumerate}
  \item $D$ is hyperimmune,
  \item $D$ is inequationally fully abstract for $\Lambda$,
  \item $D$ is fully abstract for $\Lambda$.
  \end{enumerate}
\end{em}

\section*{Conclusion}
\label{Conclusion}

In this paper, we have introduced two very new notions (hyperimmunity and quasi-approximability) on top of two known notions (full abstraction for \Hst and approximability) and a lot of different sub-notions (sensibility, extensionality, theory BT). The relations between these notions may not be clear for the reader, even for classic notions ({\em e.g.}, few people realize that full abstraction for \Hst does not implies approximability in general).

\begin{figure}
  \caption{Lattice of the properties considered in this paper.\label{fig:theories}}
  \begin{center}
    \begin{tikzpicture}
      \node (b') at (6,0.7) {$\beta$};
      \node (b) at (6,1) {$\bullet\hspace{-6.5pt}$};
      \node (H') at (2.8,1.9) {$\mathcal H$};
      \node (H) at (3,2) {$\bullet\hspace{-6.5pt}$};
      \node (bn') at (9.5,1.8) {$\beta\eta$};
      \node (bn) at (9,2)  {$\bullet\hspace{-6.5pt}$};
      \node (BT') at (3.5,2.9) {$\BT$};
      \node (BT) at (3,3)  {$\bullet\hspace{-6.5pt}$};
      %\node (Hn') at (6.07,2.6) {$\mathcal H\eta$};
      \node (Hn) at (6,3) {$\bullet\hspace{-6.5pt}$};
      %\node (BTn') at (5.6,4.1) {$\BT\eta$};
      \node (BTn) at (6,4) {$\bullet\hspace{-6.5pt}$};
      \node (app') at (-0.4,5) {$app$};
      \node (app) at (0,5) {$\bullet\hspace{-6.5pt}$};
      \node (qapp') at (2.5,4.8) {$q\dash app$};
      \node (qapp) at (3,5) {$\bullet\hspace{-6.5pt}$};
      \node (H*') at (5.7,4.9) {$\Hst$};
      \node (H*) at (6,5) {$\bullet\hspace{-6.5pt}$};
      \node (Hyp') at (9.6,5) {$Hyp$};
      \node (Hyp) at (9,5) {$\bullet\hspace{-6.5pt}$};
      \node (app2) at (1.5,7) {$\bullet\hspace{-6.5pt}$};
      \node (appn) at (3,7) {${\color{gray}\bullet}\hspace{-6.5pt}$};
      \node (qH) at (4.5,7) {$\bullet\hspace{-6.5pt}$};
      \node (HH) at (7.4,7) {$\bullet\hspace{-6.5pt}$};
      \node (HH') at (7.85,6.5) {$\bullet\hspace{-6.5pt}$};
      \node (HH'') at (8.3,6) {$\bullet\hspace{-6.5pt}$};
      \node (top) at (4.5,10) {$\bullet\hspace{-6.5pt}$};
      \draw[-,thick] (b.east) -- (H.east);
      \draw[-,thick] (b.east) -- (bn.east);
      \draw[-,thick] (H.east) -- (BT.east);
      \draw[-,thick] (H.east) -- (Hn.east);
      \draw[-,thick] (bn.east) -- (Hyp.east);
      %\draw[-,thick] (bn.east) --  (BTn.east);
      \draw[-,thick] (bn.east) --  (Hn.east);
      \draw[-,thick] (BT.east) --  (app.east);
      \draw[-,thick] (BT.east) --  (qapp.east);
      \draw[-,thick] (BT.east) -- (BTn.east);
      \draw[-,thick] (Hn.east) -- (BTn.east);
      %\draw[-,thick] (BTn.east) -- (qH.east);
      \draw[-,thick] (BTn.east) -- (H*.east);
      \draw[-,dashed,gray] (BTn.east) -- (appn.east);
      %\draw[-,thick] (app.east) --  (top.east);
      \draw[-,thick] (app.east) --  (app2.east);
      \draw[-,dashed, gray] (app.east) --  (appn.east);
      \draw[-,thick] (qapp.east) --  (app2.east);
      \draw[-,thick] (top.east) --  (app2.east);
      \draw[-,dashed,gray] (top.east) --  (appn.east);
      \draw[-,thick] (qapp.east) --  (qH.east);
      %\draw[-,thick] (qapp.east) --  (top.east);
      \draw[-,thick] (top.east) --  (qH.east);
      \draw[-,thick] (H*.east) --  (qH.east);
      \draw[-,thick] (H*.east) --  (HH.east);
      %\draw[-,thick] (H*.east) --  (top.east);
      \draw[-,thick] (top.east) --  (HH.east);
      \draw[-,thick] (HH'.east) --  (HH.east);
      \draw[-,thick] (BTn.east) --  (HH'.east);
      \draw[-,thick] (HH''.east) --  (HH'.east);
      \draw[-,thick] (Hn.east) --  (HH''.east);
      \draw[-,thick] (Hyp.east) --  (HH''.east);
      %\draw[-,thick] (Hyp.east) --  (top.east);
    \end{tikzpicture}
  \end{center}
\end{figure}
For such readers, we present, in Figure~\ref{fig:theories}, a graphic summarizing the different properties we have seen in the article. In this figure :
\begin{itemize}
\item $\beta$ stands for being a model (the name refers to the smallest $\lambda$-theory $\beta$).
\item $\mathcal H$ stands for the sensible models, {\em i.e}, those models that equate all diverging terms:
      $$ M,N\Ua \quad \Rta \quad \llb M\rrb=\llb N\rrb. $$
\item $\beta\eta$ stands for extensional models, {\em i.e}, those models preserving $\eta$-equivalence:
      $$\llb \lambda x.x\rrb = \llb \lambda xy.x\ y\rrb.  $$
\item $\BT$ stands for models that respect B\"ohm trees:
      $$ \forall M,N,\quad \BT(M)=\BT(N)\quad \Rta\quad \llb M\rrb=\llb N\rrb. $$
%\item $\BT\eta$ and $\mathcal H\eta$ simply stand for the conjunctions of $\BT$ and $\beta\eta$ and of $\mathcal H$ and $\beta\eta$.
\item \Hst stands for models that are fully abstract for \Hst:
      $$ \forall M,N,\quad M\equivob N \quad \Lra\quad \llb M\rrb = \llb N\rrb. $$
\item $app$ stands for models that are approximable:
      $$ \forall M,\quad \llb M\rrb = \llb \BT(M)\rrb_{ind}. $$
\item $q\dash app$ stands for models that are quasi-approximable:
      $$ \forall M,\quad \llb M\rrb = \llb \BT(M)\rrb_{qf}. $$
\item $Hyp$ stands for models that are hyperimmune.
\item The other nodes are simply defined as sups and do not have names.
\end{itemize}

This graphic is a lattice of properties that a K-model can satisfy, with binary sups corresponding to the conjunction of the properties (modulo logical equivalence).\footnote{Notice that two points in the graphic may well be logically equivalent.} %\footnote{With the noticeable exceptions of the conjunction $app\wedge\beta\eta$ (just below our top) which are omitted for readability.} 
In particular, one can see that quasi-approximation together with extensionality implies the full abstraction for $\Hst$. Moreover, for any among our four main properties ({\em i.e.}, $app$, \Hst, $q\dash app$ and Hyp), having any two non-adjacent properties ($app$/\Hst, $app$/$Hyp$ or $q\dash app$/Hyp) is sufficient to get the two others.

Notice that in the article we are claiming that $app$ and $q\dash app$ implies Hyperimmunity, but this was in presence of extensionality. One can then check that the sup of $app$, $q\dash app$ and $\beta\eta$ is indeed the top of our lattice.

Notice also that we placed hyperimmunity above extensionality. This is because we use extensionality in order to define hyperimmunity. A careful reader may probably be able to extend naturally hyperimmunity to a non-extensional setting, but several of the relations of Figure~\ref{fig:theories} may break with this generalization.

Finally, we conjecture that all these relations are strict in the fully general case (extended to models that are not K-models). This is proved for most already existing relations but not for the relations between $\BT$, $app$ and $\Hst$. 

Approximability is not {\em a propri} implied by $\BT$ or even by $\Hst$ but no counter-examples have been presented yet. This is a difficult question related to the characterization of sensibility. In fact it is actually difficult to get an idea of what non-approximable models lies above $\mathcal H$. Indeed, the most efficient methods we know for proving sensibility are realisability methods that are intrinsically linked with approximability~\cite{Bre16Approx}. Notice that the only result on this direction was from Kerth that created a continuum of sensible models of (disjoint) theories below \BT~\cite{Ker98}. In this paper we simply avoid the difficulty by only considering approximable K-models.

This was the first attempt at studying a $\lambda$-theory by characterising its fully abstracting models (among a relatively large class). This opens a lot of new research directions such as generalisations for larger classes of models, for other languages or for other $\lambda$-theories. The latter has actually been explored by the author in a collaborative work on Morris's extensional equivalence (the observational equivalence for weak reduction) \cite{BMPR16}. This work is bounded to relational models which are morally extensional extensions\footnote{by opposition to ``extensional collapses''.} of approximable K-models \cite{BaEh97}. There, we show that the full abstraction for Morris's equivalence corresponds to satisfy the $\lambda$-Konig property. The $\lambda$-Konig property is a sort of dual of hyperimmunity: rather than forbidding all infinite non-hyperimmune chains, it requires the presence of a dense set of such non-hyperimmune chains.

\section*{Acknowledgements}

I wish to thank Antonio Bucciarelli, Michele Pagani, Antonino Salibra and the anonymous reviewers for their proof-checking and helping in clarifying many issues. I must especially thank Michele for his time spent in the supervision of the redaction of this paper. Finally I acknowledge Thomas Ehrhard and Giulio Manzonetto for their advices and various discussions.

\bibliographystyle{plain}
\bibliography{article}

\end{document}